\documentclass[pdflatex,sn-mathphys-num]{sn-jnl}

\usepackage{graphicx}%
\usepackage{multirow}%
\usepackage{amsmath,amssymb,amsfonts}%
\usepackage{amsthm}%
\usepackage{mathrsfs}%
\usepackage[title]{appendix}%
\usepackage{xcolor}%
\usepackage{textcomp}%
\usepackage{manyfoot}%
\usepackage{booktabs}%
\usepackage{algorithm}%
\usepackage{algorithmicx}%
\usepackage{algpseudocode}%
\usepackage{listings}%
\usepackage{dsfont}
\usepackage{subfig}
\usepackage{float}
\usepackage{array,multirow,makecell}
\usepackage{amsmath,amsfonts,amssymb}
\usepackage{diagbox}
\usepackage{tabularx}
\usepackage{adjustbox}
\usepackage{svg}
\usepackage{ulem} 
\usepackage{floatpag}
\usepackage{ulem}
\floatpagestyle{empty}

\newtheorem{theorem}{Theorem}
\usepackage{pgf,tikz}

\newcommand{\proba}{\mathbb{P}}

\newcommand{\bT}{\mathbf{T}}
\newcommand{\bM}{\mathbf{M}}
\newcommand{\bY}{\mathbf{y}}
\newcommand{\bX}{\mathbf{X}}
\newcommand{\bTtilde}{\mathbf{\tilde{T}}}

\newcommand{\bMhat}{\mathbf{\hat{M}}}
\newcommand{\bYhat}{\mathbf{\hat{y}}}
\newcommand{\balpha}{\boldsymbol{\alpha}}
\newcommand{\bbeta}{\boldsymbol{\beta}}
\newcommand{\bgamma}{\boldsymbol{\gamma}}
\newcommand{\bpsi}{\boldsymbol{\psi}}
\newcommand{\bxi}{\boldsymbol{\xi}}
\newcommand{\prox}{\mbox{Prox}}

\begin{document}

\title[Group lasso based selection for high-dimensional mediation analysis]{Group lasso based selection for high-dimensional mediation analysis\footnote{This manuscript has been accepted for publication in \textit{Statistics in Medicine}.}}

\author[1,2]{\fnm{Allan} \sur{Jérolon}}\email{allan.jerolon@chu-guadeloupe.fr}

\author[2]{\fnm{Flora} \sur{Alarcon}}\email{flora.alarcon@u-paris.fr}


\author[3]{\fnm{Florence}\sur{Pittion}}\email{florence.pittion@univ-grenoble-alpes.fr}

\author[3]{\fnm{Magali}\sur{Richard}}\email{magali.richard@univ-grenoble-alpes.fr}

\author[3]{\fnm{Olivier}\sur{François}}\email{olivier.francois@univ-grenoble-alpes.fr}

\author[4]{\fnm{Etienne} \sur{Birmelé}}\email{etienne.birmele@unistra.fr}
\equalcont{These authors contributed equally to this work.}

\author[2]{\fnm{Vittorio} \sur{Perduca}}\email{vittorio.perduca@u-paris.fr}
\equalcont{These authors contributed equally to this work.}

\affil[1]{Centre d'Investigation Clinique Antilles Guyane, Inserm CIC 1424, CHU de
Guadeloupe, Les Abymes, Guadeloupe, France}

\affil[2]{Université Paris Cité, CNRS, MAP5, F-75006 Paris, France}

\affil[3]{Univ. Grenoble Alpes, CNRS, UMR 5525, VetAgro Sup, Grenoble INP, TIMC, 38000 Grenoble, France}

\affil[4]{Institut de Recherche Mathématique Avancée, UMR 7501 Université de Strasbourg et CNRS, 
7 rue René-Descartes, 67000 Strasbourg, France}

\abstract{Mediation analysis aims to identify and estimate the effect of an exposure on an outcome that is mediated through one or more intermediate variables. 
In the presence of multiple intermediate variables, two pertinent methodological questions arise: estimating mediated effects when mediators are correlated, and performing high-dimensional mediation analyses when the number of mediators exceeds the sample size. This paper presents a two-step procedure for high-dimensional mediation analyses. The first step selects a reduced number of candidate mediators using an ad-hoc lasso penalty. The second step applies a procedure 
we previously developed to estimate the mediated effects, accounting for the correlation structure among the retained candidate mediators. We compare the performance of the proposed two-step procedure with state-of-the-art methods using simulated data. Additionally, we demonstrate its practical application by estimating the causal role of DNA methylation in the pathway between smoking and rheumatoid arthritis using real data.}


\maketitle

\section{Introduction}\label{sec1}

Mediation analyses methods are widely used in biomedical and social sciences to disentangle the causal effect of a treatment on an outcome through intermediate variables called mediators. Modern causal mediation analysis is based on counterfactual variables and aims at decomposing the total effect into a direct effect and the mediated, or indirect, effect(s) carried by the mediator(s) \cite{pearl_direct_2001, imai_general_2010}.

In many practical problems, for instance in biomedical applications with intermediate variables of genomic nature, the number of potential mediators exceeds the sample size, leading to the high-dimensional mediation problem. Several methods have been proposed in recent years to address this challenging problem, for a review of the literature see \cite{blum2020challenges, han2023mediation}. Existing methods can be broadly categorized into two main families based on their approach to dimensionality reduction.

Methods in the first family build uncorrelated linear combinations of potentials mediators, using PCA \cite{huang_hypothesis_2016}, sparse PCA \cite{han_sparse_2020} or PLS \cite{assi_statistical_2015} approaches. In \cite{chen_high-dimensional_2018} a linear combinations of candidate mediators is chosen by maximising a criterion based on the joint likelihood of the treatment/mediator and mediator/outcome models. This approach is extended in \cite{geuter2020multiple} using a generalized version of population value decomposition (PVD).
With any of these methods, the mediated effect carried by each linear combination can be evaluated, and the weights of the mediators within these linear combinations reveal their contribution to the mediated effects.

A second family of approaches, to which this paper belongs, involves screening the candidate mediators to select a subset and subsequently estimating their mediated effects. \cite{van_kesteren_exploratory_2019} proposes to explore the set of possible mediators by a coordinate descent updating at each step the status of a small number of potential mediators. \cite{derkach_high_2019} reduces the dimensionality by introducing a small set of latent variables governing both the potential mediators and the outcome. To introduce further approaches, let us assume linear (or logistic) regression models, and let $\balpha$ be the vector of the coefficients of the exposure in the regression models of the candidate mediators given the exposure (one model per mediator), and $\bbeta$ the vector of the coefficients of the candidate mediators in the model of the outcome given the mediators and the exposure. With these notations, a third way to select mediators is to suppose that $\balpha$ and $\bbeta$ follow Gaussian mixture models whose base distributions are centered and with either small or large variance. 
\cite{song_bayesian_2018} proposes a Bayesian Sparse Linear Mixed Model for high-dimensional mediation analysis which is a one-step method. In contrast, the HDMAX2 method \cite{jumentier2023high} makes no distributional assumption. For each candidate mediator $M_k$, the HDMAX2 method tests the nullity of $\alpha_k$ and $\beta_k$, and the squared maximum of the two corresponding p-values is considered as a new p-value used as a selection criterion.  \cite{djordjilovic2019global, dai2022multiple} similarly consider the maximum of the p-values and introduce testing procedures that allow to control the FDR. \cite{dai2024controlling} also achieves to control the global FDR, but rather considers a two-step procedure controlling the FWER on $\balpha$ and the FDR on $\bbeta$.  

Other methods for the selection of mediators rely on penalized likelihood optimization with the selection method varying according to the considered model and penalization. After reducing the pool of candidate mediators from a large number to a moderate number by employing the sure independence screening, \cite{zhang_estimating_2016}, and its extension \cite{perera2022hima2}, conduct variable selection with the minimax concave penalty, or a de-biased lasso procedure respectively, and finally carry out joint significance testing for mediation effect. Interestingly, \cite{loh_non-linear_2020} considers a different definition of the mediated effect, called interventional indirect effect, that needs less stringent hypothesis on the joint law of the mediators. The selection strategy relies  on two penalized regression, for $\balpha$ and $\bbeta$, respectively.

In this article, we propose a new two-step approach for the selection of candidate mediators and the estimation of individual mediated effects. The first filtering step reduces the number of candidate mediators by solving a penalized optimization problem with group lasso penalty that takes simultaneously the parameters of interest $\balpha$ and $\bbeta$ into account. 
Once the number of candidate mediators is lower than the sample size, the second step consists in running the algorithm developed in \cite{jerolon_causal_2020} to estimate and test the mediated effects of the retained mediators, together with the direct effect.

This article is organized as follows. Section~\ref{sect_mediation} defines the problem of high-dimensional mediation analysis and introduces the notations and underlying hypotheses. Our algorithm is detailed in Section~\ref{sect_method}. The results of the comparisons with previously published methods on synthetic datasets are reported in Section~\ref{sect_simulation}. An illustration on a real dataset is shown in Section~\ref{sect_application}. Section~\ref{sect_discussion} discusses our results.  

\begin{center}
\begin{figure}[htp]
\begin{tikzpicture}
\node (m1)	at (7,9) 	{$M_1$};
\node (m2)	at (7,8.5) 	{$M_2$};
\node (m3)	at (7,8) 	{$M_3$};
\node (m4)	at (7,7) 		{$M_4$};
\node (m5)	at (7,6.5) 	{$M_5$};
\node (m6)	at (7,5.5) 	{$M_6$};
\node (m7)	at (7,5) 		{$M_7$};
\node (mK9)	at (7,4) 	{$M_{14}$};
\node (mK10)	at (7,3.5) 	{\vdots};
\node (mK13)	at (7,3) 		{$M_{K}$};
\node (t1) 		at (0,6) 	{$T$};
\node (y) 		at (14,6) 	{$Y$};

edges %
\begin{scope}[every path/.style={->, line width= 1}, every node/.style={sloped, inner sep=1pt}]
\draw (t1) --  (m1)  ;
\draw (t1) --  (m2)  ;
\draw (t1) --  (m3)  ;
\draw (t1) --  (m4)  ;
\draw (t1) --  (m5)  ;

\draw (m1) --  (y)  ;
\draw (m2) --  (y)  ;
\draw (m3) --  (y)  ;
\draw (m6) --  (y)  ;
\draw (m7) --  (y)  ;


\end{scope}
\end{tikzpicture} 
\caption{Example of a high-dimensional mediation model with treatment $T$ and outcome $Y$. The direct effect of $T$ on $Y$ and the effects of pretreatment confounders $X$ on all depicted variables are included in the model but omitted from the figure for readability. Candidate mediators $M_1$ to $M_3$ are true mediators, while $M_4$ to $M_K$ are not. \label{plusTbigmed}}
\end{figure}
\end{center}
 

\section{A high-dimensional mediation analysis model}\label{sect_mediation}

We consider a mediation model with a binary exposure (or treatment) $T$, $K$ candidate mediators $(M_1,\ldots,M_K)$ and an outcome $Y$. Let $(X_1,\dots,X_L)$ be the vector of pretreatment confounders. An example is shown in Figure~\ref{plusTbigmed}.  
If $K$ is large, in particular larger than the sample size $n$, the problem of identifying and inferring the direct and mediated effects in the model is referred to as \textit{high-dimensional mediation analysis}. In this high-dimensional setting, the aim of our algorithm is to identify which candidate mediators truly have a mediated effect and to estimate the corresponding direct and mediated effects. Both the candidate mediators and the outcome are assumed to be either Gaussian or binary, and are therefore modeled using either Gaussian or logistic regression models, respectively. We consider the following data structures:

\begin{itemize}
 \item $n\times 1$ column vector $\bT$, where the entry $t_i$ is the $i^{th}$ observation of $T$ 
 \item $n\times K$ matrix $\bM$, where the entry $m_{ik}$ is the $i^{th}$ observation of $M_k$
 \item $n \times 1$ column vector $\bY$, where the entry $y_i$ is the $i^{th}$ observation of $Y$
 \item $n\times L$ matrix $\bX$, where the entry $x_{il}$ is the $i^{th}$ observation of 
$X_l$. 
\end{itemize}

\subsubsection*{Regression models for the candidate mediators $M_k$}

If the $k^{th}$ canidate mediator is continuous, we assume the following Gaussian model:
\[   M_k = \alpha_{0k} + \alpha_{1k} T + \sum_{l=1}^L \xi_{lk} X_l + \epsilon_k     \mbox{ with } \epsilon_k\sim \mathcal{N}(0,\sigma_k^2). \]
We denote  by $\hat{m}_{ik}(\balpha,\bxi)$ the associated prediction for the $i^{th}$ individual seen as a function of the model parameters: $\hat{m}_{ik}(\balpha,\bxi) = \alpha_{0k} + \alpha_{1k} t_{i} + \sum_{l=1}^L \xi_{lk} x_{il}$. 

If the $k^{th}$ potential mediator is binary, we assume the following logistic regression model:
\[  \log \left(\frac{\proba(M_k=1)}{1-\proba(M_k=1)}\right) = \alpha_{0k} +  \alpha_{1k} T  + \sum_{l=1}^L \xi_{lk} X_l. \]
We then denote $\hat{m}_{ik}(\balpha,\bxi)$ the associated prediction for $\proba(m_{ik}=1)$, that is 
$\hat{m}_{ik} = \exp(\hat{\nu}_{ik})/(1+\exp(\hat{\nu}_{ik}))$ with $\hat{\nu}_{ik}(\balpha,\bxi) = \alpha_{0k} + \alpha_{1k} t_{i}+ \sum_{l=1}^L \xi_{lk} x_{il}$. All predictions $\hat{m}_{ik}$ are compiled into the matrix $\bMhat(\balpha,\bxi)$. 

\subsubsection*{Regression model for the outcome $Y$}

If the outcome is continuous, we assume the following Gaussian model
\[   Y = \gamma_{0} + \gamma_{1} T + \sum_{k=1}^K \beta_{k} M_k + \sum_{l=1}^L \psi_{l} X_l+ \epsilon     \mbox{ with } \epsilon \sim \mathcal{N}(0,\sigma^2). \]
We denote $\hat{y}_{i}(\bbeta,\bgamma,\bpsi)$ the prediction for the $i^{th}$ individual: $\hat{y}_{i} (\bbeta,\bgamma,\bpsi)= \gamma_{0} + \gamma_{1} t_{i} + \sum_{k=1}^K \beta_{k} \hat{m}_{ik} + \sum_{l=1}^L \psi_{l} x_{il}.$ 
If the outcome is binary, we consider the following logistic model
\[  \log\left(\frac{\proba(Y=1)}{1-\proba(Y=1)}\right) = \gamma_{0} + \gamma_{1} T + \sum_{k=1}^K \beta_{k} M_k  + \sum_{l=1}^L \psi_{l} X_l.\]
In this case, 
$\hat{y}_{i} (\bbeta,\bgamma,\bpsi) = \exp(\hat{z}_{i})/(1+\exp(\hat{z}_{i}))$
with $\hat{z}_{i} = \gamma_{0} + \gamma_{1} t_{i} + \sum_{k=1}^K \beta_{k} \hat{m}_{ik} + \sum_{l=1}^L \psi_{l} x_{il}$. In both cases, all predictions $\hat{y}_i$ are compiled into the vector $\bYhat(\bbeta,\bgamma,\bpsi)$. In applications, predictions are obtained by substituting the parameters with their estimated values, typically calculated by maximising the likelihood. This yields the estimated matrices $\bMhat(\hat{\balpha},\hat{\bxi})$ and $\bYhat(\hat{\bbeta},\hat{\bgamma},\hat{\bpsi})$.

\section{MAHI: a two-step algorithm for Mediation Analysis with HIgh-dimensional data}\label{sect_method}

The method proposed in this paper is a two-step procedure based on a previous work~ \cite{jerolon_causal_2020} that allows, in a small-dimensional setting (with fewer candidate mediators than observations), to estimate a confidence interval for the mediated effect of each candidate mediator, even when they are (uncausally) related. 
The goal of the first step in our MAHI algorithm is to select a subset of $K_{\max}<n$ candidate mediators. This step aims to reduce the dimensionality of the problem while retaining as many true mediators as possible.
In the second step, the aforementioned previously developed method is applied to refine this selection by excluding candidate mediators that exhibit no significant mediated effect. 

\subsection{Step 1: from high to low dimension}

 This step relies on an ad hoc loss function depending on the parameters $(\balpha,\bbeta,\bgamma,\bxi,\bpsi)$ and on a group lasso procedure with stability selection.

\subsubsection*{Definition of the loss functions}

We consider the following loss functions for the regression models of the candidate mediators and the outcome: 
\[  \ell_{M_k}(\balpha,\bxi) = \left\{ \begin{array}{ll}
  \frac{1}{2} \sum_{i=1}^n (\hat{m}_{ik}-m_{ik})^2 & \mbox{if } M_k \mbox{ is Gaussian} \\
  \sum_{i=1}^n -m_{ik} \hat{\nu}_{ik} + \log(1+\exp(\hat{\nu}_{ik}))  & \mbox{if } M_k \mbox{ is binary} 
  \end{array} \right. \]
and
\[  \ell_Y(\bbeta,\bgamma,\bpsi) = \left\{ \begin{array}{ll}
  \frac{1}{2} \sum_{i=1}^n (\hat{y}_{i}-y_{i})^2 & \mbox{if } Y \mbox{ is Gaussian} \\
  \sum_{i=1}^n -y_{i} \hat{z}_{i} + \log(1+\exp(\hat{z}_{i})) & \mbox{if } Y \mbox{ is binary}. 
  \end{array} \right. \]
The loss function associated to the whole model is then defined as
\[    f(\balpha,\bbeta,\bgamma,\bxi,\bpsi) = \frac{1}{n} \left(\sum_{k=1}^K \ell_{M_k}(\balpha,\bxi) + w_Y\ell_Y(\bbeta,\bgamma,\bpsi)\right), \]
where the weight $w_Y>0$ allows to tune the relative importance of the treatment-outcome and mediators-outcome relationships. 
Varying $w_Y$ will therefore allow to select candidate mediators with different behaviors.

\subsubsection*{The group lasso and the proximal operator}

The group lasso \cite{yuan_model_2006,meier_group_2008} is used to select candidate mediators by minimizing a penalized version of $f$, with a penalty that promotes sparsity by encouraging the nullity of some pre-defined groups of parameters. More precisely,  we group the coefficients $\alpha_{1k}$ and $\beta_k$ for each candidate mediator $M_k$ in order to jointly select them either out of the model (false mediators) or into the model (promising candidate mediators that deserve further inspection). The considered problem can then be written, for a given regularization parameter $\lambda>0$, as
\begin{equation}
    \underset{\balpha,\bbeta,\bgamma,\bxi,\bpsi}{\text{argmin}} f(\balpha,\bbeta,\bgamma,\bxi,\bpsi) + \lambda \sum_{k=1}^K \sqrt{\alpha_{1k}^2+\beta_k^2}.  
\label{eq:optimization_pr}    
\end{equation}
To solve this optimization problem, we employ the proximal method as described in \cite{bach2011optimization}. This method relies on an iterative procedure described in Appendix~\ref{sec:theor_det}.

We emphasize that the goal of Step 1 is to transform the initial problem into a small-dimensional problem while discharging as few true mediators as possible. The selection of some false mediators is not problematic because Step 2 will test individual mediated effects and exclude those candidate mediators $M_k$ whose mediated effects are not significant. A crucial consequence of this approach is that, under the assumption that the number of real mediators of interest is less than the sample size $n$, the value of the penalty parameter $\lambda$ does not require fine-tuning and can be chosen to be relatively small. This makes it possible to retain as many true mediators as possible in the list of candidate mediators selected by Step 1, with the only requirement being that the size of the selection reduces the original problem to a low-dimensional mediation analysis. 

\subsubsection*{Stability selection and parameter choices}

Lasso selection is known to be highly unstable, in the sense that small perturbations in the data may significantly change the selection. This problem can be addressed using the stability selection procedure introduced in \cite{meinshausen_stability_2010}, which selects variables based on the number of times they are chosen when running the original selection procedure on $N_{\text{boot}}$ bootstrap samples. The underlying idea is that relevant variables should, despite the instability, be selected more often than non-relevant ones. Moreover, as described earlier, using different values of $w_Y$ may enable the capture of mediators with values for $\alpha_{1k}$ and $\beta_k$ varying in different scale ranges.

Based on all these considerations, we propose the following procedure for the first selection step: 

\begin{itemize}

\item Specify a grid of values for $w_Y$, and integers $N_{\text{boot}}$ and $K_{\max}$.  

\item For each value of $w_Y$:
\begin{itemize}
    \item choose a value of $\lambda$ by dichotomy such that the number of retained candidate mediators is in a pre-defined interval, by default $[n/2,n]$.
    \item Obtain $N_{\text{boot}}$ bootstrap samples from the original data. For each bootstrap sample, run the proximal method to solve the optimization problem~(\ref{eq:optimization_pr}).
\end{itemize}


\item Rank the candidate mediators based on their frequency of selection across all obtained lists, from most to least frequently selected. The underlying rationale is that a true mediator should be selected more frequently than a non-mediating variable, which may only be selected occasionally as a false positive. 
\item Select the top $K_{\max}$ ranked candidate mediators.
\end{itemize}

The choice of $K_{\max}$ is guided by the fact that Step 2 is based on estimating the parameters of ``classic'' (i.e., non-penalized) regression models and that the number of the explanatory variables has to be chosen accordingly.  For a continuous outcome the default value is $2n/\log(n)$. 
For a binary outcome $K_{\max}$ must at most be equal to the integer part of $-2+n/50$ 
according to \cite{sample_2018}.
\medskip

It has to be noted that the soft thresholding induced by the lasso may select  candidate mediators for which only one of the coefficients $\alpha_{1k}$ or $\beta_k$ is non-null. These false mediators will be dealt with in Step 2. However, in a very high-dimensional setting, the number of such candidates may saturate the number $K_{\max}$. We then suggest using a pre-filtering step using the first step of HDMAX2 \cite{pittion2025published} applying a high p-value threshold. This test is specifically designed to eliminate the candidate mediator when either  $\alpha_{1k}=0$ or $\beta_k=0$. We applied this strategy in the application on real data presented in Section \ref{sect_application}.

\subsection{Step 2: estimation of direct and mediated effects}

The second step of MAHI involves estimating and testing the mediated effects through each of the selected candidate mediators, which we denote $M_1,\ldots,M_{K_{\max}}$ (up to a permutation of the original indices). 

In particular, we adopt the definition of the average indirect effect through $M_k$ defined in \cite{imai_identification_2013}. For each treatment value $t\in\{0,1\}$, we consider the average difference in counterfactual outcomes  
\begin{eqnarray}
\delta^k(t) =  \mathbb{E}\left[Y(t,M_k(1),W_k(t))\right] - \mathbb{E}\left[Y(t,M_k(0),W_k(t))\right], \label{dk}
\end{eqnarray}
where $W_k$ denotes the vector of all candidate mediators except $M_k$, and $W_k(t)$ and $M_k(t)$ denote counterfactual variables. For simplicity, we focus on the average effect $\delta^k = (\delta^k(1)+\delta^k(0))/2$. Estimating and testing these effects relies on the identifiability assumptions and method for low-dimensional multiple mediation analysis described in \cite{jerolon_causal_2020}. For clarity, we sketch here the corresponding quasi-Bayesian algorithm adapted from \cite{imai_general_2010} and refer the reader to \cite{jerolon_causal_2020} for all the details and its theoretical justification. 

\bigskip

\textbf{Algorithm for low-dimensional multiple mediation analysis:}

\begin{enumerate}
    \item Fit parametric models for the outcome and the $K_{\max}$ retained candidate mediators.
    \item Simulate $J$ times the model parameters from their estimated Gaussian multivariate sampling distribution.
    \item For each simulation, repeat the following steps:
    \begin{itemize}
        \item For each individual, simulate the vector of counterfactual candidate mediators. 
        \item For each individual, simulate the counterfactual outcomes corresponding to the simulated values of the counterfactual candidate mediators. 
        \item For each candidate mediator, calculate the mediated effect by averaging over all individuals.
    \end{itemize}
    \item For each candidate mediator, from the empirical distribution obtained above, calculate the point estimate of the mediated effect together with p-values and confidence intervals.
\end{enumerate}
In applications, we suggest to set $J=1000$, as in \cite{jerolon_causal_2020,JSSv059i05}. The final selection of mediators consists of the set of candidate mediators whose confidence intervals do not contain $0$ after correction for the $K_{\max}$ multiple comparisons. The type of multiple correction is left to the user. Note that, as detailed in  \cite{jerolon_causal_2020}, this algorithm also allows for the estimation of the direct and joint mediated effects.

\section{Simulation study}\label{sect_simulation}

We ran simulations to validate MAHI and to compare it to methods recently introduced in the literature.

\subsection{Models for simulated data}

\subsubsection{Continuous outcome}

We conducted two simulation studies with $100$ replicates each. The first study involved independent candidate mediators, while the second study considered correlated candidate mediators. In each replicate, we included $n = 100$ observations and $K = 500$ candidate mediators, simulated according to the model

\begin{equation}
\begin{array}{ccl}
M_{ik} & = & \mu_k  + \alpha_k T_i + \xi_{1k} X_{1i} + \xi_{2k} X_{2i} + \epsilon_{ik}\\ 
Y_i & = & 20 + 50 T_i + \sum_{k=1}^{K} \beta_k M_{ik} + \psi_{1k} X_{1i} + \psi_{2k} X_{2i} + \epsilon_{i0}
\end{array}\label{grand_simed}
\end{equation}
where $1\leq i\leq n$ and $1\leq k\leq K$. The exposure variable $T$ follows a Bernoulli distribution, $T \sim \mathcal{B}(0.3)$, and for each $k$, $\mu_k$ is drawn uniformly in the interval $[-2,2]$. Table~\ref{grand_param} shows the values of $\alpha_k$ and $\beta_k$ for the first 50 variables $M_k$. The higher the absolute value of $\alpha_k\beta_k$, the greater the mediated effect through $M_k$. As such, the first 10 mediators have strong mediated effects (and are, in principle, easier to select), the next 10 have mild mediated effects (less easy to detect) and the next 10 have weak mediated effects (hard to detect). All other 470 variables $M_k$ are not true mediators because $\alpha_k=0$ or $\beta_k=0$. 
In the first simulation study, the $\epsilon_{ik}$ are i.i.d. according to $\mathcal{N}(0,1)$ and $\xi_{1k} = \xi_{2k} = \psi_{1k} = \psi_{2k} = 0$ for all $k$. In the second study, we considered 10 clusters of candidate mediators, where within each cluster, variables are correlated, and the clusters themselves are independent. The first cluster includes $(M_1, M_{11}, M_{21}, M_{31}, M_{41}, M_{51},\dots,M_{491})$, the second cluster includes $(M_2, M_{12}, M_{22}, M_{32}, M_{42}, M_{52},\dots,M_{492})$, and so on. Thus, each cluster consists of 50 variables including one strong mediator, one mild mediator, one weak mediator and 47 non-mediating variables. Within each cluster the $\epsilon_{ik}$ were simulated according to a centered multivariate normal distribution with all variances equal to 1 and pairwise correlations set to 0.9. Covariates $X_{1}$
and $X_{2}$ follow a standard normal distribution. The values of $\xi_{1k}, \xi_{2k},\psi_{1k}$ and $\psi_{2k}$ are taken between 1 and 6.
\begin{table}[htp]
\begin{tabular}{|c|cccccccccc|}
\hline
$\boldsymbol{k}$ & \bf{1} & \bf{2} & \bf{3} & \bf{4} & \bf{5} & \bf{6} & \bf{7} & \bf{8} & \bf{9} & \bf{10} \\
\hline
$\alpha_{k}$ & -15 & 6 & -6 & -13 & 11 & 16 & -9 & 9 & 14 & 20\\
$\beta_{k}$ & -11 & 11 & -15 & -5 & 7 & 13 & 14 & -7 & -9 & 11\\
\hline
\hline
$\boldsymbol{k}$ & \bf{11} & \bf{12} & \bf{13} & \bf{14} & \bf{15} & \bf{16} & \bf{17} & \bf{18} & \bf{19} & \bf{20} \\
\hline
$\alpha_{k}$ & -3 & 2 & -2 & 4 & -2 & -1 & 4 & 1 & -1 & 2\\
$\beta_{k}$ & 4 & 2 & 2 & -1 & 3 & 2 & -3 & 3 & 3 & 3\\
\hline
\hline
$\boldsymbol{k}$ & \bf{21} & \bf{22} & \bf{23} & \bf{24} & \bf{25} & \bf{26} & \bf{27} & \bf{28} & \bf{29} & \bf{30} \\
\hline
$\alpha_{k}$  & 0.5 & 0.2 & -0.5 & -0.3 & 0.7 & -0.3 & 0.8 & -0.2 & 0.6 & -0.3\\
$\beta_{k}$  & 0.5 & 0.2 & -0.7 & 0.7 & 0.6 & -0.6 & -0.6 & 0.2 & -0.7 & 0.3\\
\hline
\hline
$\boldsymbol{k}$& \bf{31} & \bf{32} & \bf{33} & \bf{34} & \bf{35} & \bf{36} & \bf{37} & \bf{38} & \bf{39} & \bf{40} \\
\hline
$\alpha_{k}$ & 20 & 20 & 20 & 20 & 20 & 4 & 4 & 4 & 4 & 4\\
$\beta_{k}$ & 0 & 0 & 0 & 0 & 0 & 0 & 0 & 0 & 0 & 0\\
\hline
\hline
$\boldsymbol{k}$ & \bf{41} & \bf{42} & \bf{43} & \bf{44} & \bf{45} & \bf{46} & \bf{47} & \bf{48} & \bf{49} & \bf{50} \\
\hline
$\alpha_{k}$  & 0 & 0 & 0 & 0 & 0 & 0 & 0 & 0 & 0 & 0\\
$\beta_{k}$  & 20 & 20 & 20 & 20 & 20 & 4 & 4 & 4 & 4 & 4\\
\hline
\end{tabular}\caption{
Values of $\alpha_k$ and $\beta_k$ for $k=1,\ldots,50$. For $k=51,\ldots,500$, $\alpha_k=\beta_k=0$.} \label{grand_param}
\end{table}

\subsubsection{Binary outcome}

We simulated $100$ replicates, including $n = 1350$ observations and $K = 2000$ independent candidate mediators each, according to the model
\begin{equation}
\begin{array}{ccl}
M_{ik} & = & 1  + \alpha_k T_i + \epsilon_{ik}\\ 
Y_i^* & = & -65 +  T_i + \sum_k \beta_k M_{ik} +\epsilon_{i0}\\
Y_i & = & \mathds{1}_{Y_i^* > 0}
\end{array}\label{bin_grand_simed}
\end{equation}
where $1\leq i\leq n$ and $1\leq k\leq K$. The exposure variable $T$ follows a Bernoulli distribution, $T \sim \mathcal{B}(0.3)$, the residuals $\epsilon_{i0}$ follow a logistic distribution,  $\epsilon_{i0} \sim \mathcal{L}(0,1)$, and  the $\epsilon_{ik}$ are i.i.d. according to  $\mathcal{N}(0,1)$ for each $k$. As shown in Table \ref{bin_grand_param}, the $15$ true mediators $M_1,\ldots,M_{15}$ are split in three sets of 5 mediators each, with strong, mild and weak mediated effects respectively. Using Monte Carlo simulations, we determined that these parameter values result in average mediated effects of 0.078, 0.018, and 0.004 for the strong, mild, and weak mediators, respectively. 

\begin{table}[htp]
\begin{tabular}{|c|cccccccccc|}
\hline
$\boldsymbol{k}$ & \bf{1} & \bf{2} & \bf{3} & \bf{4} & \bf{5} & \bf{6} & \bf{7} & \bf{8} & \bf{9} & \bf{10} \\
\hline
$\alpha_{k}$ & 2 & 2 & 2 & 2 & 2 & 1 & 1 & 1 & 1 & 1\\
$\beta_{k}$ & 2 & 2 & 2 & 2 & 2 & 1 & 1 & 1 & 1 & 1\\
\hline
\hline
$\boldsymbol{k}$ & \bf{11} & \bf{12} & \bf{13} & \bf{14} & \bf{15} & \bf{16} & \bf{17} & \bf{18} & \bf{19} & \bf{20} \\
\hline
$\alpha_{k}$ & 0.5 & 0.5 & 0.5 & 0.5 & 0.5 & 0.5 & 0.5& 0.5 & 0.5 & 0.5\\
$\beta_{k}$ & 0.5 & 0.5 & 0.5 & 0.5 & 0.5 & 0& 0 & 0 & 0 & 0\\
\hline
\hline
$\boldsymbol{k}$ & \bf{21} & \bf{22} & \bf{23} & \bf{24} & \bf{25} & \bf{26} & \bf{27} & \bf{28} & \bf{29} & \bf{30} \\
\hline
$\alpha_{k}$  &0 & 0 & 0 & 0 & 0 & 0 &0 & 0 & 0 & 0\\
$\beta_{k}$  & 5 & 5 & 5 & 5 & 5 & 0 & 0 & 0 & 0 & 0\\
\hline
\end{tabular}\caption{Values of $\alpha_k$ and $\beta_k$ for $k=1,\ldots,30$. For $k=31,\ldots,2000$, $\alpha_k=\beta_k=0$.}\label{bin_grand_param}
\end{table}

\subsection{Methods settings}

We implemented our method in the \texttt{mahi} function of the GitHub R package \texttt{AllanJe/mahi}. For our simulation studies, we considered $N_{\text{boot}} = 30$, $w_Y=(1,2,\ldots,7,8)$, and $J=1000$. We set $K_{\max}$ to $30$ and $25$ for the analyses of the data with continuous and binary outcomes, respectively. 
This constraint ensures that the second step no longer deals with a high-dimensional setting. For the second step, p-values were adjusted using the Benjamini-Hochberg correction with a false discovery rate (FDR) threshold of $0.20$.

\subsubsection{Comparison to state-of-the-art methods for continuous outcomes}

We compared MAHI to the following six alternative methods on simulated data with continuous outcomes:

\begin{enumerate}

\item \cite{van_kesteren_exploratory_2019} introduced an approach for high-dimensional mediation analysis, called the \textit{Coordinate-wise Mediation Filter} (\underline{CMF}). The CMF implementation consists of two components: an internal algorithm which performs the selection of mediators by coordinate descent using a decision function $D$, and an external algorithm that runs several times the internal algorithm and aggregates the corresponding outputs. The entire procedure is implemented in the GitHub R package \texttt{vankesteren/cmfilter}. In our simulations, the decision function is the Sobel test. The external algorithm is run $1000$ times. Once the selection rate for each mediator is calculated, a mediator is chosen if its selection rate is greater than $0.079$, the value recommended by the authors.

\item \cite{zhang_estimating_2016} introduced the \underline{HIMA} (\textit{HIgh-dimensional Mediation Analysis}) algorithm, which is based on penalized regressions and uses a lasso-type penalty function called the concave minimax penalty (MCP)~\cite{zhang_nearly_2010}. The HIMA implementation consists of three steps: first the set of candidate mediators is reduced by means of the sure independent screening (SIS) method \cite{fan_sure_2008}, then the estimates $\widehat{\beta_k}$ are calculated using the MCP penalization criterion, and at last mediated effects are tested and p-values are adjusted according to the Bonferroni correction. The entire procedure is implemented and available in the R package \texttt{hima}. In our simulations, we chose the first $n/\log(n)$ mediators obtained with the SIS method, as recommended by the authors. 

\item \cite{gao_testing_2019} proposed a variation of HIMA allowing the selection of correlated candidate mediators, called \underline{HDMA} (\textit{High-Dimensional Mediation Analysis}). The HDMA method differs from HIMA in the second step, where debiased estimates of $\widehat{\beta_k}$ are calculated. The entire procedure is available in the GitHub R package \texttt{YuzhaoGao/High-dimensional-mediation-analysis-R}. In our simulations the settings are the same as for HIMA. 

\item \cite{song_bayesian_2018} introduced the \underline{BAMA} (\textit{Bayesian Mediation Analysis}) approach. It is a Bayesian inference method using continuous shrinkage priors to extend previous causal mediation analyses techniques to a high-dimensional setting.  For each candidate mediator, the posterior inclusion probability (PIP) is estimated measuring the association strength between exposure and mediators and between mediators and outcome. The candidate mediators with the highest PIP are selected as the active mediators. The entire procedure is implemented and available in the R package \texttt{bama}. The performance of BAMA is critically dependent on a user-specified PIP threshold. In our simulations we chose a PIP threshold of 0.1.

\item \cite{zhao_sparse_2020} introduced the \underline{SPCMA} (\textit{Sparse Principal Component Mediation Analysis}) algorithm. When candidate mediators are potentially causally related to one another, one approach is to perform a principal component analysis (PCA) to obtain orthogonal principal components (PCs), which can be treated as new, conditionally independent mediators. However, these new candidate mediators, which are linear combinations of the original candidate mediators, can be difficult to interpret. The sparse high-dimensional mediation analysis approach proposed in~\cite{zhao_sparse_2020} applies PCA with sparse loadings, making the principal components more interpretable as they are linear combinations of a few original candidate mediators. The entire procedure is implemented in the GitHub R package \texttt{zhaoyi1026/spcma}. In our simulations, variables $M_k$ are causally independent so we used the function recommended by the authors in this case, 
which performs marginal causal mediation analysis under the linear structural equation modeling framework.

\item \cite{jumentier_high-dimensional_nodate, pittion2025published} introduced the \underline{HDMAX2} procedure (\textit{High Dimensional mediation analysis with $\max^2$ test}). 
The selection procedure of HDMAX2 involves fitting latent factor mixed models (LFMMs,~\cite{caye2019lfmm}) to estimate the effects of exposure on mediators and the effect of each mediator on the outcome. For each candidate mediator, two p-values ($P_x$ and $P_y$) are derived from these models, testing the null hypotheses of no effect of exposure on the mediator and no effect of the mediator on the outcome, respectively. Candidate mediators are then selected using the $\max^2$ test, a novel test that uses the p-value $p=\max\{P_x, P_y\}^2$. Similar to the Sobel test, the $\max^2$ test rejects the null hypothesis that either the effect of exposure on the mediator or the effect of the mediator on the outcome is null. The selected candidate mediators are subsequently ranked by significance, and only those below a given threshold proceed to the second step. This step consists of performing simple mediation analyses for each selected candidate mediator using the \texttt{mediation} package~\cite{JSSv059i05} to estimate and test their mediated effects. The threshold can be determined using data-adaptive approaches, such as FDR control, or set manually by the user. In our study, we retained the 50 candidate mediators with the lowest $\max^2$ p-values. HDMAX2 is available in the GitHub R package \texttt{bcm-uga/hdmax2}.

\end{enumerate}

\subsubsection{Comparison to state-of-the-art methods for binary outcomes}

We compared MAHI to HIMA, HDMA and HDMAX2, all of which can also be applied to binary outcomes. After the first step of MAHI, we retained the top 
$\left\lfloor \frac{n}{50} - 2 \right\rfloor$ candidate mediators to proceed to the second step. For the three other methods we proceeded as follows :

\begin{enumerate}
\item For HIMA, we chose the 
first $\lceil n/(2\log(n))\rceil$ candidate mediators obtained with the SIS method, as recommended by the authors for a binary outcome.
\item For HDMA, we also chose the 
first $\lceil n/(2\log(n))\rceil$ mediators obtained with the SIS method, as recommended by the authors for a binary outcome.
\item For HDMAX2, we retained the top 25 candidate mediators at the end of the first step to proceed to the second step. We then applied the Hochberg correction to the results of the second step at a threshold of 0.05.
\end{enumerate}

Note that the implementations of HIMA and HDMA allow to choose different penalisation methods to obtain sparsity. We run them all, which explains the multiple results for each of the methods in Table~\ref{tab_true_binary}.

\subsection{Results}

Table \ref{tab_true_normal} and Table \ref{tab_true_binary} show, for each method, the mean of three performance metrics, namely precision, recall and specificity, over 100 replicates, for continuous and binary outcomes respectively. We recall that precision, or positive predictive value, is the proportion of 
variables of interest among those selected by the method;
recall, or sensitivity, is the proportion of 
selected variables among those of interest;
and specificity is the proportion of 
variables not selected among those of no interest. In particular, these metrics are defined with respect to four selection problems:

\begin{itemize}
\item the selection of all true mediators, 
\item the selection of strong mediators, 
\item the selection of mild mediators, 
\item the selection of weak mediators. 
\end{itemize}

Figures~\ref{res_true_normal} and~\ref{res_true_binary}  show the distribution of the three metrics across 100 replicates with independent candidate mediators, for continuous and binary outcomes, respectively. Figure~\ref{res_true_normal_corr} shows the distribution of the three metrics for the model with correlated mediators and continuous outcomes, and covariates included.
Figures~\ref{res_fake_normal} and~\ref{res_fake_binary} show the distribution of the false discovery rate (1-precision), the false negative rate (1-recall) and the false positive rate (1-specificity) across replicates with independent candidate mediators for continuous and binary outcomes, respectively. Figure~\ref{res_fake_normal_corr} displays the distribution of these three metrics across replicates with correlated mediators and continuous outcomes for the model that includes covariates. More specifically, the Figures~\ref{res_fake_normal},~\ref{res_fake_binary} and \ref{res_fake_normal_corr} pertain to the following selection problems:
\begin{itemize}
\item the selection of false mediators, 
\item the selection of false mediators with $\alpha_k \neq 0$ and $\beta_k =0$, 
\item the selection of false mediators with $\alpha_k = 0$ and $\beta_k \neq 0$.
\end{itemize}

Table \ref{tab:CP_ind} shows additional results for continuous candidate mediators, including the average number of selections by Step 1 and Step 2, the average bias, and the average coverage of the confidence intervals (calculated at Step 2) for the true indirect effect values.

\subsubsection{Results, continuous outcomes}

\begin{table}[htp]
\centering
\begin{tabular}{|c|c|ccc|ccc|}
\hline
\multicolumn{2}{|c|}{} & \multicolumn{3}{c|}{Independent candidate mediators} & \multicolumn{3}{c|}{Correlated candidate mediators}\\
\multicolumn{2}{|c|}{} & \multicolumn{3}{c|}{without covariates} & \multicolumn{3}{c|}{with covariates}\\
\hline
& Method			& Precision				& Recall				& Specificity & Precision				& Recall				& Specificity\\
\hline
\rotatebox{ 90}{ \hbox to 15pt{\hss All true mediators} } & MAHI & \textbf{0.897} & 0.225 & \textbf{0.996} & 0.901 & 0.108 & 0.999 \\ 
 & CMF & 0.556 & 0.109 & 0.994 & - & - & - \\
 & HIMA & 0.140 & 0.119 & 0.958 & \textbf{0.902} & 0.015 & \textbf{1.000} \\
 & HDMA & 0.790 & 0.193 & \textbf{0.996} & 0.218 & 0.029 & 0.992 \\  
 & BAMA & 0.666 & \textbf{0.608} & 0.981 & 0.608 & 0.177 & 0.993 \\ 
 & MCMA & 0.645 & 0.053 & 0.958 & - & - & - \\ 
 & HDMAX2 & 0.873 & 0.052 & 0.987 & 0.216 & \textbf{0.222} & 0.929 \\  
\hline
\rotatebox{ 90}{ \hbox to 15pt{\hss Strong mediators} } & MAHI & 0.816 & 0.595 & \textbf{0.996} & \textbf{0.855} & 0.306 & 0.999 \\ 
  & CMF & 0.321 & 0.185 & 0.992 & - & - & - \\
  & HIMA & 0.102 & 0.266 & 0.958 & 0.685 & 0.034 & \textbf{1.000} \\
  & HDMA & 0.676 & 0.494 & 0.995 & 0.111 & 0.041 & 0.991 \\ 
  & BAMA & 0.343 & \textbf{0.937} & 0.963 & 0.327 & 0.281 & 0.988 \\
  & MCMA & 0.630 & 0.045 & 0.986 & - & - & - \\ 
  & HDMAX2 & \textbf{0.824} & 0.114 & 0.995 & 0.118 & \textbf{0.349} & 0.926 \\
\hline
\rotatebox{ 90}{ \hbox to 19pt{\hss Mild mediators} } & MAHI & 0.083 & 0.066 & 0.985 & 0.046 & 0.018 & 0.993 \\ 
  & CMF & 0.205 & 0.124 & 0.990 & - & - & - \\
  & HIMA & 0.025 & 0.058 & 0.953 & 0.174 & 0.008 & \textbf{0.999} \\ 
  & HDMA & 0.101 & 0.075 & 0.986 & 0.030 & 0.012 & 0.991 \\ 
  & BAMA & 0.303 & \textbf{0.831} & 0.961 & \textbf{0.249} & 0.210 & 0.987  \\
  & MCMA & \textbf{0.365} & 0.016 & 0.985 & - & - & - \\ 
  & HDMAX2 & 0.158 & 0.011 & \textbf{0.993} & 0.083 & \textbf{0.256} & 0.924 \\
\hline
\rotatebox{ 90}{ \hbox to 10pt{\hss Weak mediators} } & MAHI & 0.028 & 0.012 & 0.984 & 0.000 & 0.000 & 0.992 \\
& CMF & 0.030 & 0.019 & 0.988 & - & - & - \\
& HIMA & 0.013 & 0.032 & 0.953 & 0.043 & 0.002 & \textbf{0.999} \\
& HDMA & 0.013 & 0.011 & 0.985 & \textbf{0.077} & 0.033 & 0.991 \\
& BAMA & 0.020 & \textbf{0.057} & 0.945 & 0.032 & 0.041 & 0.983 \\ 
& MCMA & \textbf{0.350} & 0.014 & 0.985 & - & - & - \\
& HDMAX2 & 0.110 & 0.004 & \textbf{0.993} & 0.015 & \textbf{0.061} & 0.920 \\ 
\hline
\end{tabular}
\caption{Comparison of high-dimensional mediation analysis methods with regards to the ability to select the true mediators $M_1,\ldots,M_{30}$: mean precision, recall and specificity over the 100 replicates simulated with \textbf{continuous} outcomes according to model \eqref{grand_simed}.}\label{tab_true_normal} 
\end{table}

Table~\ref{tab_true_normal} presents the results for both simulation settings with a continuous outcome, noting that CMF and MCMA are not designed to handle covariates.
In the simplest setting (independent mediators, no covariates), MAHI achieved the highest overall precision and a competitive recall compared to all other methods except BAMA. 
In the more challenging and realistic setting (correlated mediators, covariates included), MAHI maintained high precision, comparable to the top-performing method HIMA and greater than that of all other methods. Notably, HDMAX2, which performed similarly in the first setting, was outperformed in this scenario. Although MAHI’s recall degraded, the decline was moderate.
Overall, MAHI emerged as a competitive choice in terms of the precision-recall trade-off.


A deeper examination of the results, separating strong, mild, and weak mediators, indicates that MAHI had a great ability at selecting the strongest mediators but performed poorly in detecting the weakest ones. However, this behavior is reasonable when dealing with high-dimensional problems, such as mediation through the methylation of CpG sites, where candidate mediators are numerous and locally correlated. In such applications, focusing on the strongest signals is both practical and justified.

Table \ref{tab:CP_ind} shows that when either $\alpha_{1k}$ is large and $\beta_k = 0$, or, to a lesser extent, $\alpha_{1k} = 0$ and $\beta_k$ is large, Step 1 tended to select the corresponding false mediators. However, Step 2 effectively filtered out the vast majority of these false positives. Additionally, while the empirical coverage of the confidence intervals remained close to the nominal level for independent candidate mediators, it was lower for some of the correlated candidate mediators.


\medskip

\subsubsection{Results, binary outcome}

\begin{table}[htp]
\centering
\begin{tabular}{|c|c|ccc|}
\hline
& Method			& Precision				& Recall				& Specificity\\
\hline
\rotatebox{ 90}{ \hbox to 15pt{\hss All true mediators} } 
& MAHI & \bf{0.992} & \bf{0.767} & \bf{1.000} \\ 
& HIMA\_lasso & 0.760 & 0.153 & 0.999 \\ 
& HIMA\_MCP & 0.742 & 0.304 & 0.999 \\ 
& HIMA\_SCAD & 0.687 & 0.225 & 0.998 \\ 
& HDMA\_lasso & 0.717 & 0.619 & 0.998 \\
& HDMA\_ridge & 0.709 & 0.522 & 0.998 \\
& HDMAX2 & 0.991 & 0.295 & \bf{1.000} \\ 
\hline
\rotatebox{ 90}{ \hbox to 15pt{\hss Strong mediators} } 
& MAHI & 0.438 & \bf{1.000} & 0.997 \\ 
& HIMA\_lasso & 0.583 & 0.324 & 0.999 \\ 
& HIMA\_MCP & 0.547 & 0.622 & 0.998 \\ 
& HIMA\_SCAD & 0.500 & 0.430 & 0.998 \\ 
& HDMA\_lasso & 0.390 & 0.994 & 0.996 \\ 
& HDMA\_ridge & 0.448 & 0.962 & 0.997 \\ 
& HDMAX2 & \bf{0.837} & 0.732 & \bf{1.000} \\
\hline
\rotatebox{ 90}{ \hbox to 19pt{\hss Mild mediators} } 
& MAHI & \bf{0.406} & \bf{0.936} & 0.997 \\ 
& HIMA\_lasso & 0.180 & 0.102 & 0.999 \\ 
& HIMA\_MCP & 0.173 & 0.224 & 0.997 \\ 
& HIMA\_SCAD & 0.206 & 0.180 & 0.998 \\ 
& HDMA\_lasso & 0.241 & 0.632 & 0.995 \\ 
& HDMA\_ridge & 0.198 & 0.456 & 0.995 \\ 
& HDMAX2 & 0.126 & 0.128 & 0.998  \\ 
\hline
\rotatebox{ 90}{ \hbox to 10pt{\hss Weak mediators} } 
& MAHI & \bf{0.148} & \bf{0.366} & 0.995 \\ 
& HIMA\_lasso & 0.097 & 0.030 & 0.998 \\ 
& HIMA\_MCP & 0.072 & 0.066 & 0.997 \\ 
& HIMA\_SCAD & 0.131 & 0.062 & 0.997 \\ 
& HDMA\_lasso & 0.085 & 0.230 & 0.994 \\ 
& HDMA\_ridge & 0.064 & 0.148 & 0.995 \\ 
& HDMAX2 & 0.028 & 0.026 & 0.998 \\
\hline
\end{tabular}
\caption{Comparison of high-dimensional mediation analysis methods with regards to the ability to select the true mediators $M_1,\ldots,M_{15}$: mean precision, recall and specificity over the 100 replicates simulated with \textbf{binary} outcomes according to model \eqref{bin_grand_simed}.}\label{tab_true_binary} 
\end{table}

Table~\ref{tab_true_binary} demonstrates that MAHI achieved the best recall, with an average of only 23\% of the true mediators not being selected, and the best precision, with almost all of the selected variables being true mediators. While MAHI successfully selected all strong mediators, its precision was lower compared to almost all concurrent methods. However, when selecting mild and weak mediators, MAHI exhibited the best recall and precision.


\section{Illustration on real data : mediation of smoking on rheumatoid arthritis outcomes}
\label{sect_application}

\subsection{Biological context}

Rheumatoid Arthritis (RA) is a chronic inflammatory disease influenced by both genetic and environmental factors. Smoking has been identified as one of the most important extrinsic risk factor for its development and severity \cite{Chang2014-gv}. DNA methylation (DNAm),  an epigenetic mechanism that involves the methylation of specific bases in the DNA strand, can regulate gene transcription, thereby affecting disease development. The relationship  between DNAm levels and RA occurence was first investigated in \cite{Liu2013-bv}. In addition, several association studies have already established the impact of tobacco consumption on DNAm \cite{Kaur2019-cz}. As a case study, we explored to which extent DNAm mediates the effect of tobacco consumption on the occurrence of RA. The dataset was collected from the Gene Expression Omnibus (GEO) database using the accession number GSE42861 \cite{Liu2013-bv}. It consists of Illumina HumanMethylation450 BeadChip array in peripheral blood leukocytes (PBLs) from RA patients (n = 354) and normal controls (n = 333). Clinical data including age, gender, smoking status and residential area were provided for each sample. Two patients were excluded from the analysis because their smoking status was unknown.

\medskip

\subsection{Mediation analysis}

To proceed with the mediation analysis, the categorical smoking status variable was transformed into a binary variable. Patients who had never smoked or reported only occasional smoking were classified as non-smokers (coded as 0). Former and current smokers were grouped together and classified as smokers (coded as 1). Additionally, age and gender were included as adjustment variables in the model. The DNAm matrix included 473,864 CpG probes (i.e. features) across 687 patients. 
Due to the very large initial number of probes,
a preliminary selection was done using the HDMAX2 method.
First, we used the \textit{hdmax2\_step1} approach to run association studies for all potential mediators and to test the significance of the estimated mediated effects. Then we applied a filter to select the top $1000$ probes with the most significant p-values (Figure~\ref{man_and_forest_plot}A). The resulting subset of DNAm probes is still high-dimensional but computationally less expensive. Subsequently, the MAHI method was applied to this refined subset of DNAm probes with the tuning parameters set to $N_{\text{boot}}=50, w_Y=(1,2,\ldots,7,8)$, $K_{\max}=50$, and $J=1000$. 
Mediated ORs, corresponding to the indirect effect mediated by DNAm probes, were estimated for the selected subset of CpGs along with their CI. The top 50 CpGs mediators are depicted in Figure~\ref{man_and_forest_plot}B.

\begin{figure}[htp]
    \centering
    \includegraphics[width=\textwidth]{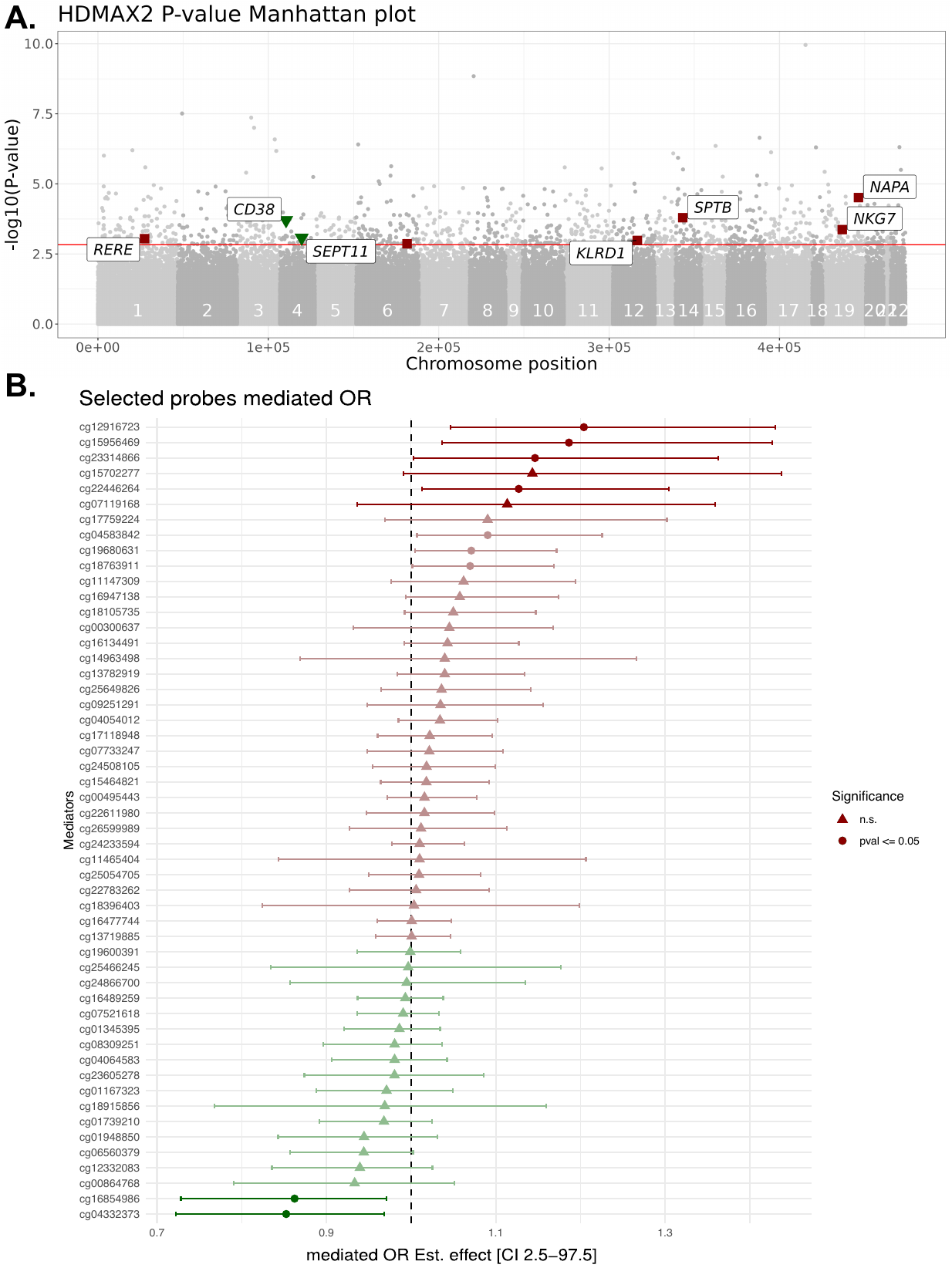}
    \caption{Summary of mediation analysis of smoking on RA occurrence through DNA methylation. \textbf{A} Manhattan plot displaying the –log10 transformed p-values estimated using the max-squared method (HDMAX2) for each CpG site. Each dot represents an individual CpG, ordered on the $x$-axis according to their genomic position. The red line indicates the threshold for the top $1,000$ CpGs selected for further analysis, on which MAHI was applied. Red squares represent probes with MAHI ORs greater than $1.10$, while green triangles represent probes with MAHI ORs lower than $0.9$.  Labels correspond to genes associated with the selected probes, if any. Chromosome numbers are labeled in white.   \textbf{B} Mediated ORs for the top 50 mediators. The estimate effect is represented by a dot and its unadjusted 95\% CI by the bar. Symbols correspond to the significance cut off of 5\%  (square for p-value $\geq 0.05$, circle p-value $< 0.05$). Colors correspond to the sign and importance of the effect (dark green for estimated OR under  $0.9$, light green for estimated OR between $0.9$ and $1$, pink for estimated OR between $1$ and $1.1$ and dark red for estimated OR over $1.1$). 
    }
    \label{man_and_forest_plot}
\end{figure}


\subsection{Biological interpretation}

Table~\ref{probes_info} summarizes the results and relevant biological information for the selected CpG mediators that show ORs greater than $1.10$ and lower than $0.9$. When OR values are lower than $0.9$, occurrence of RA is significantly reduced. In this context, our method identified two CpG mediators  (cg04332373 and cg16854986) for which methylation appears  to decrease in RA patients.
When OR values are greater than $1.1$, the occurrence of RA is significantly increased. Interestingly, we observe varying scenarios in terms of mediated effects for ORs $\geq 1.1$. In some instances, the methylation of CpG mediators decreases in RA patients compared to controls (e.g. cg23314866, cg15702277 and cg22446264), while in other cases, it increases (e.g. cg07119168, cg12916723 and cg15956469). This illustrates complex mediation pathways, suggesting that different biological processes are likely at play. We also examined whether some genes associated with the selected CpG mediators were previously known in the literature to be linked to RA (Table~\ref{probes_info}, ``Pubmed hits'' column). Interestingly, our approach not only identified known candidates (i.e. \textit{CD38}) but also discovered new probes that had not previously been associated with RA, opening the way to new research perspectives and experimental validation. \\

\begin{table}[htp]
\renewcommand{\arraystretch}{1.2}
\centering
\begin{adjustbox}{max width=\textwidth}
\begin{tabular}
{p{1.2cm}c>{\centering\arraybackslash}p{1.5cm}>{\centering\arraybackslash}p{1.5cm}p{0.7cm}>{\centering\arraybackslash}p{1.4cm}>{\centering\arraybackslash}p{2.2cm}}

\textbf{CpG Probes} & \textbf{mediated OR} & \textbf{mean DNAm cases} & \textbf{mean DNAm controls} & \textbf{Chr} & \textbf{Associated genes} &   \textbf{Pubmed hits} 

\\ \hline

\multicolumn{7}{c}{\textbf{OR less than 0.90}} 

\\ \hline

cg04332373 & 0.85{[}0.72, 0.97{]}** & 0.20 $\pm$ 0.03 & 0.22 $\pm$ 0.03 & chr4 & \textit{CD38} &   147 

\\

cg16854986 & 0.86{[}0.73, 0.97{]}** & 0.12 $\pm$ 0.03 & 0.13 $\pm$ 0.04 & chr4 & \textit{SEPT11} &   0

\\ \hline

\multicolumn{7}{c}{\textbf{OR more than 1.10}} 

\\ \hline

cg23314866 & 1.15{[}1.00, 1.36{]}** & 0.24 $\pm$ 0.05 & 0.29 $\pm$ 0.04 & chr19 & \textit{NAPA} &   1 

\\

cg07119168 & 1.11{[}0.94, 1.35{]} & 0.81 $\pm$ 0.03 & 0.78 $\pm$ 0.04 & chr14 & \textit{SPTB} &    3

\\

cg12916723 & 1.20{[}1.04, 1.43{]}*** & 0.63 $\pm$ 0.03 & 0.61 $\pm$ 0.04 & chr19 & \textit{NKG7} &   2

\\

cg15702277 & 1.14{[}0.99, 1.44{]}* & 0.25 $\pm$ 0.05 & 0.31 $\pm$ 0.05 & chr1 & \textit{RERE} &   0 

\\

cg15956469 & 1.18{[}1.04, 1.43{]}** & 0.89 $\pm$ 0.04 & 0.85 $\pm$ 0.05 & chr12 & \textit{KLRD1} &   8 

\\

cg22446264 & 1.13{[}1.01, 1.30{]}** & 0.47 $\pm$ 0.07 & 0.54 $\pm$ 0.07 & chr6 & - &   - \\ \hline

\end{tabular}
\end{adjustbox}
\caption{For each selected probes: mediated OR (with CI, *, **, *** , res. significant OR at $5\%$, $1\%$ and $0.1\%$ type I error), DNAm mean $\pm$ standard deviation for cases group and controls group, chromosome in which probe is located, nearest gene (identified using Illumina annotations), 
and the number of Pubmed matching hits with gene symbol and RA.}
\label{probes_info}
\end{table}

\section{Discussion and conclusion}
\label{sect_discussion}

In this article we introduced MAHI, a two step-procedure for high-dimensional mediation analysis where the candidate intermediate variables outnumbers the available observations. In Step 1, MAHI first performs variable selection in the pool of candidate mediators through a group lasso penalty that we adapted specifically to the mediation problem. Then, in Step 2, MAHI estimates and tests the direct and mediated causal effects in the resulting lower-dimensional mediation model using the multiple mediation analysis method we developed in \cite{jerolon_causal_2020}.

On simulated data, MAHI generally achieved good results compared to alternative methods. Specifically, it outperformed existing methods in terms of precision, recall, and specificity when applied to binary outcomes. On simulated data with continuous outcomes, MAHI demonstrated the best overall precision and its recall, while not the highest, ranked in the middle among the assessed methods, reflecting a solid balance between precision and recall. In particular, MAHI had good  performances with correlated candidate mediators when compared to the other methods. This is a appealing attribute of MAHI, as in practical applications, candidate mediators are often expected to be correlated, typically due to residual confounding. It is, however, important to stress that the performance of MAHI declined with mild and weak mediators.


The principal methodological novelty of this work is Step 1 of MAHI. Our simulation results suggest that integrating this initial step with our previously developed inferential algorithm yields highly satisfactory performance. However, it is important to note that Step 1 can, in principle, be implemented prior to any method designed for low-dimensional analysis.  Nevertheless, when handling correlated candidate mediators, we suggest following through with the second step of MAHI, as detailed in this article. 

One limitation of the current implementation of our two-step procedure is that the indirect effects of candidate mediators are tested in Step 2 using the same data used for their selection in Step 1. While our empirical study suggests that this post-selection inference has a limited impact on precision and recall, likely because strong mediators dominate the signal, it remains a potential source of bias. A simple way to mitigate this limitation would be to split the data into two parts, using one subset for selection and the other for inference. However, this requires a sufficiently large initial dataset to maintain statistical power.

When dealing with an extremely large number of candidate mediators, such as hundreds of thousands, the current R implementation of MAHI may become computationally ineffective. The complexity of MAHI is largely influenced by the choice of the number of bootstrap replicates $N_{\text{boot}}$, the length of the grid of weights $w_Y$ for the first step, as well as the number of simulations $J$ for the quasi-Bayesian algorithm in the second step. There is substantial potential to parallelize the current implementation with respect to these parameters, which would greatly enhance its execution speed. This is left for future works. Another possibility in presence of an extremely large number of candidate mediators is to run mediator pre-selection with the fast first step of the HDMAX2 approach.

We employed the strategy combining the first step of HDMAX2 followed by MAHI to detect and assess the role of DNA CpG site methylation in mediating the impact of smoking on the occurrence of rheumatoid arthritis and  identified $8$ significant probes. 
Remarkably, one of the $8$ selected probes was associated with the \textit{CD38} gene, which shows a strong association with RA in PubMed research, with 147 hits. CD38 is important in the regulation of innate immunity \cite{ye_potential_2023} and has already been identified as a potential therapeutical target for autoimmune diseases such as RA, but also systemic lupus or multiple sclerosis \cite{Peclat2020-ua}.

An interesting feature of the multiple mediation framework presented in this article is that it is possible to extend the loss functions considered in Sections~\ref{sect_method} to incorporate user-defined partitions of the candidate mediators. This extension would allow for the integration of existing domain knowledge into the model. Moreover, it would also be possible to include multiple treatments. By adapting the first step of MAHI to this more general case, the selection process could be enhanced to favor groups of mediators that exhibit a common mediated effect across all treatments. The ability to incorporate existing knowledge about the structure of candidate mediators would be especially valuable in genomic applications, where the focus is frequently on evaluating the mediated effects of specific genomic regions. Note that \cite{djordjilovic_global_2019} had already proposed a multiple testing procedure to determine which groups of variables had a significant mediating effect. However, such a MAHI extension would be to our knowledge the first screening method capable of taking group structure into account, as well as considering several treatments simultaneously and promoting the selection of common mediators. 
This interesting features would allow to select candidate mediators with mediated effects with respect to all exposures and to discharge intermediate variables that act as mediators only with respect to some of the exposures. However, in this situation the interpretation of the set of coefficients $\balpha$ and $\bbeta$ in terms of mediated effect is not straightforward and needs further investigation.  We leave this extension and its validation through simulation studies for future works.

Several methodological questions remain open and constitute challenging tasks for the future. Notably, it would be interesting to adapt MAHI to other types of data, in particular to longitudinal data and/or survival models. A second major question concerns the robustness of the method to violations of the conditional independence conditions, which are crucial for the identification of mediated effects (see, for instance, \cite{jerolon_causal_2020}). To the best of our knowledge, such a sensitivity analysis framework has not yet been developed in the context of high-dimensional mediation analysis.

\subsubsection*{Method availability}
The MAHI method is available in its beta version as an R package at
\texttt{https://github.com/AllanJe/mahi}.

\subsubsection*{Declaration of generative AI and AI-assisted technologies in the writing process}
During the preparation of this work the authors used Le Chat - Mistral AI in order to revise the language of some portions of the manuscript. After using this tool, the authors reviewed and edited the content as needed and take full responsibility for the content of the article.

\bibliography{bibliog.bib}

\newpage

\begin{appendices}

\section{Theoretical details}\label{sec:theor_det}

We describe how we solve the optimization problem~(\ref{eq:optimization_pr}) with the proximal method. This method can be written, with
$v=(\balpha,\bbeta,\bgamma,\bxi,\bpsi)$ and $\Omega(v)=\sum_{k=1}^K \sqrt{\alpha_{1k}^2+\beta_k^2}$, as

\[  v^{t+1} = \underset{v}{\text{argmin}}  \quad f(v^t) + \langle\nabla f(v^t), v-v^t \rangle + \lambda \Omega(v) + \frac{L}{2} \|  v-v^t \|_2^2  \]

for a well-chosen $L$. It can also be rewritten as
\begin{align*}
  v^{t+1} & = \underset{v}{\text{argmin}}  \quad \frac{1}{2} \left\lVert v- \left(v^t -\frac{1}{L} \nabla f(v^t)\right) \right\rVert_2^2 + \frac{\lambda}{L} \Omega(v) \\
         &= \prox_{\frac{\lambda}{L}\Omega} \left(v^t -\frac{1}{L} \nabla f(v^t) \right)
\end{align*}

where the proximal operator is defined as
\[ \prox_{\mu \Omega} (u) =  \underset{v}{\text{argmin}}  \quad \frac{1}{2} \|v- u \|_2^2 + \mu \Omega(v). \]

When $\Omega$ is a group lasso penalty, the proximal operator is known.  In the present case, $v^{t+1}$ is obtained by replacing, for each $k=1,\ldots,K$, the coordinates of $v^t$ corresponding to $(\alpha_{1k}, \beta_{k})$ by 
$$
\max\left\{0,\left(1-\frac{\mu}{\|(\alpha_{1k},\beta_k)\|_2}\right)(\alpha_{1k}, \beta_{k})\right\}
$$

The choice of $L$ is again made according to \cite{bach2011optimization} by increasing it until the former proximal solution verifies
\[ f(v^{t+1})  \leq f(v^t) + \langle \nabla f(v^t), v^{t+1}-v^t \rangle + \frac{L}{2} \|  v^{t+1}-v^t \|_2^2. \]

\subsubsection*{Computing the gradient}

In order to run the proximal method to select a subset of candidate mediators, the only step still needed is to compute the gradient of the loss function, which is easily done by the following result.

\begin{theorem}
  Let  $\nabla_{\balpha} f$ (respectively $\nabla_{\bxi} f$) be the matrix regrouping all the partial derivatives $\frac{\partial f}{\partial \alpha_{pk}}$ (respectively $\frac{\partial f}{\partial \xi_{lk}}$). Similarly, denote by $\nabla_{\bbeta} f$,
  $\nabla_{\bgamma} f$ and $\nabla_{\bpsi} f$ the partial gradients relative to the $\beta_k$ , the $\gamma_p$ and the $\psi_l$ coefficients. Finally, let $\bTtilde$ be the matrix obtained by adding a column of $1$'s on the left of $\bT$ (i.e., with a slight abuse of notation, we introduce $\tilde{t}_{ip}$ such that, for all $1 \leq i \leq n$, $\tilde{t}_{i0}=1$ and $\tilde{t}_{i1}=t_{i1}$).  Then

  \begin{align*}
    \nabla_{\balpha} f &=  \frac{1}{n} \bTtilde' (\bMhat(\balpha,\bxi)-\bM)    \\
    \nabla_{\bxi} f &=  \frac{1}{n} \bX' (\bMhat(\balpha,\bxi)-\bM)    \\
    \nabla_{\bbeta} f  &=  \frac{w_Y}{n} \bM' (\bYhat(\bbeta,\bgamma,\bpsi) - \bY)  \\
    \nabla_{\bgamma} f &=  \frac{w_Y}{n} \bTtilde' (\bYhat(\bbeta,\bgamma,\bpsi) - \bY) \\
    \nabla_{\bpsi} f &=  \frac{w_Y}{n} \bX' (\bYhat(\bbeta,\bgamma,\bpsi) - \bY). 
  \end{align*}
  
\end{theorem}

\begin{proof}[Proof of Theorem 1]

   $\balpha$ and $\bxi$ play symmetric roles in the mediator models, whether the Gaussian or logistic model is chosen. It is therefore sufficient to prove the equalities for $\balpha$ and the same result holds for $\bxi$ by changing $\bTtilde$ into $\bX$. The same holds for $\bgamma$ on one hand and $\bbeta$ and $\bpsi$ on the other hand, by changing $\bTtilde$ into $\bM$ and $\bX$ respectively. Only the proofs for $\balpha$ and $\bgamma$ are therefore fully developed. Their adaptation to $\bbeta$, $\bxi$ and $\bpsi$ are straightforward.

  Consider $k$ such that $M_k$ is gaussian. Then, for every $p\in\{0,1\}$,
  \begin{align*}
    \frac{\partial f}{\partial \alpha_{pk}} & = \frac{1}{n} \frac{\partial \ell_{M_k}}{\partial \alpha_{pk}} \\
                                           & = \frac{1}{2n}  \frac{\partial}{\partial \alpha_{pk}} \left( \sum_{i=1}^n \left(\sum_{q=0}^1 \alpha_{qk} \tilde{t}_{iq} + \sum_{l=1}^L \xi_{lk} x_{il} - m_{ik}\right)^2 \right) \\
                                           & = \frac{1}{n}  \sum_{i=1}^n \tilde{t}_{ip} \left(\sum_{q=0}^1 \alpha_{qk} \tilde{t}_{iq} + \sum_{l=1}^L \xi_{lk} x_{il} - m_{ik}\right)\\
                                           &=  \frac{1}{n}  \sum_{i=1}^n \tilde{t}_{ip} (\hat{m}_{ik} - m_{ik})\\
                                           &= \frac{1}{n}  \big(   \bTtilde' (\bMhat-\bM)  \big)_{pk}. 
  \end{align*} 

  The same reasoning applies when $k$ is such that $M_k$ is binary:
   \begin{align*}
    \frac{\partial f}{\partial \alpha_{pk}} & = \frac{1}{n} \frac{\partial \ell_{M_k}}{\partial \alpha_{pk}} \\
                                           & = \frac{1}{n}  \frac{\partial}{\partial \alpha_{pk}} \left( \sum_{i=1}^n -m_{ik}\left( \sum_{q=0}^1 \alpha_{qk} \tilde{t}_{iq}  + \sum_{l=1}^L \xi_{lk} x_{il}\right)+ \log\left(1+e^{\sum_{q=0}^1 \alpha_{qk} \tilde{t}_{iq} + \sum_{l=1}^L \xi_{lk} x_{il}}\right) \right) \\
                                           & = \frac{1}{n}  \sum_{i=1}^n \left(-\tilde{t}_{ip} m_{ik} + \frac{\tilde{t}_{ip}e^{\sum_{q=0}^1 \alpha_{qk} \tilde{t}_{iq}+ \sum_{l=1}^L \xi_{lk} x_{il}}}{1+e^{\sum_{q=0}^1 \alpha_{qk} \tilde{t}_{iq}+ \sum_{l=1}^L \xi_{lk} x_{il}}}\right)\\
                                           &=  \frac{1}{n} \sum_{i=1}^n \tilde{t}_{ip} (\hat{m}_{ik} - m_{ik})\\
                                           &= \frac{1}{n}  \big(   \bTtilde' (\bMhat-\bM)  \big)_{pk}. 
  \end{align*} 

   The claim concerning $\nabla_{\balpha} f$ is therefore true.
   \smallskip
   
    Let us now consider $Y$ to be Gaussian. 
   For every $p\in\{0,1\}$,
   \begin{align*}
    \frac{\partial f}{\partial \gamma_{p}} & = \frac{w_Y}{n} \frac{\partial \ell_{Y}}{\partial \gamma_{p}} \\
                                           & = \frac{w_Y}{2n} \frac{\partial}{\partial \gamma_{p}} \left( \sum_{i=1}^n \left(\sum_{q=0}^1 \gamma_{q} \tilde{t}_{iq} + \sum_{k=1}^K \beta_k m_{ik} + \sum_{l=1}^L \psi_{l} x_{il}- y_{i}\right)^2 \right) \\
                                           & = \frac{w_Y}{n} \sum_{i=1}^n \tilde{t}_{ip} \left(\sum_{q=0}^1 \gamma_{q} \tilde{t}_{iq} + \sum_{k=1}^K \beta_k m_{ik} + \sum_{l=1}^L \psi_{l} x_{il}- y_{i}\right)\\
                                           &=  \frac{w_Y}{n} \sum_{i=1}^n \tilde{t}_{ip} (\hat{y}_{i} - y_{i})\\
                                           &= \frac{w_Y}{n}  \big(   \bTtilde' (\bYhat-\bY)  \big)_{p} .
  \end{align*} 
   
   In the case of a binary outcome,
  
  \begin{align*}
    \frac{\partial f}{\partial \gamma_{p}} & = \frac{w_Y}{n} \frac{\partial \ell_{Y}}{\partial \gamma_{p}} \\
                                   \begin{split}
                                           & = \frac{w_Y}{n}  \frac{\partial}{\partial \gamma_{p}}  \left(\sum_{i=1}^n -y_{i}\left( \sum_{q=0}^1 \gamma_{q} \tilde{t}_{iq} + \sum_{l=1}^K \beta_l m_{il} + \sum_{l=1}^L \psi_{l} x_{il}\right)+\right.\\
                                             & \qquad + \left. \log\left(1+e^{\sum_{q=0}^1 \gamma_{q} \tilde{t}_{iq} + \sum_{l=1}^K \beta_l m_{il}+ \sum_{l=1}^L \psi_{l} x_{il}}\right) \right)
                                    \end{split}
                                    \\
                                             & =
                                             \frac{w_Y}{n} \sum_{i=1}^n \left(-y_{i} \tilde{t}_{ip} + \frac{\tilde{t}_{ip}e^{\sum_{q=0}^1 \gamma_{q} \tilde{t}_{iq}+\sum_{l=1}^K \beta_l m_{il}+ \sum_{l=1}^L \psi_{l} x_{il}}}{1+e^{\sum_{q=0}^1 \gamma_{q} \tilde{t}_{iq}+\sum_{l=1}^K \beta_l m_{il}+ \sum_{l=1}^L \psi_{l} x_{il}}}\right)\\
                                           &=  \frac{w_Y}{n} \sum_{i=1}^n \tilde{t}_{ip} (\hat{y}_{i} - y_{i})\\
                                           &= \frac{w_Y}{n}  \big(   \bTtilde' (\bYhat-\bY)  \big)_{p} .
  \end{align*} 

  The claims on 
  $\nabla_{\bgamma} f$ are therefore true in both cases.  
\end{proof}

\section{Additional simulation results}\label{sec:supp_figures}

\begin{figure}[htp]
\vspace{-.5cm}
\centering
\includegraphics[height=0.24\textwidth,width=0.3\textwidth]{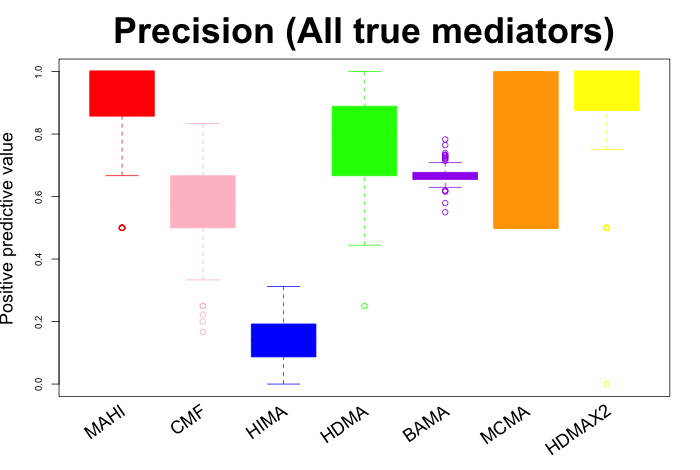}
\includegraphics[height=0.24\textwidth,width=0.3\textwidth]{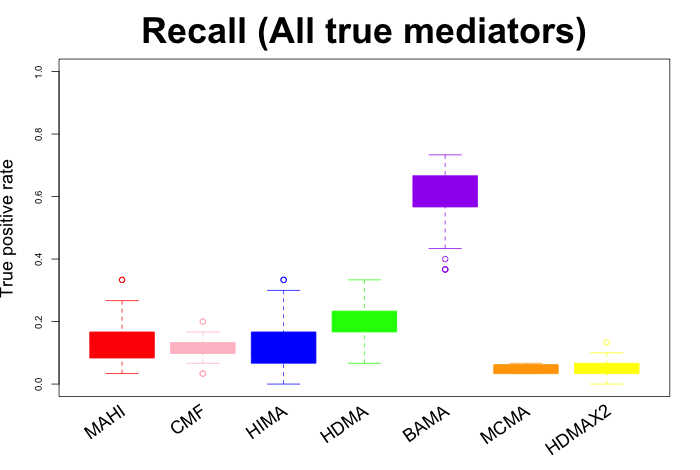}
\includegraphics[height=0.24\textwidth,width=0.3\textwidth]{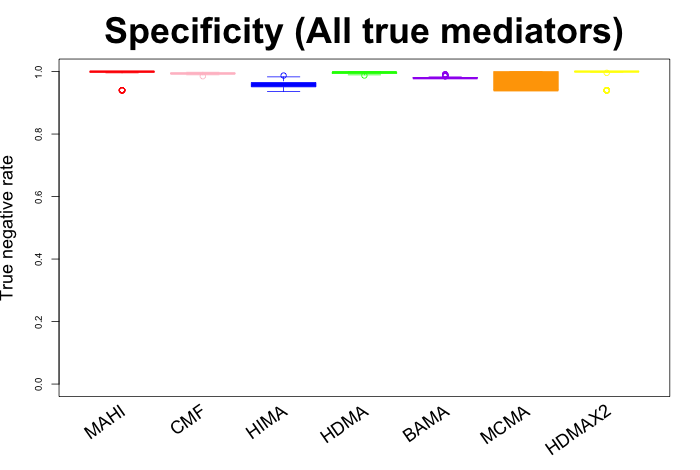}
\includegraphics[height=0.24\textwidth,width=0.3\textwidth]{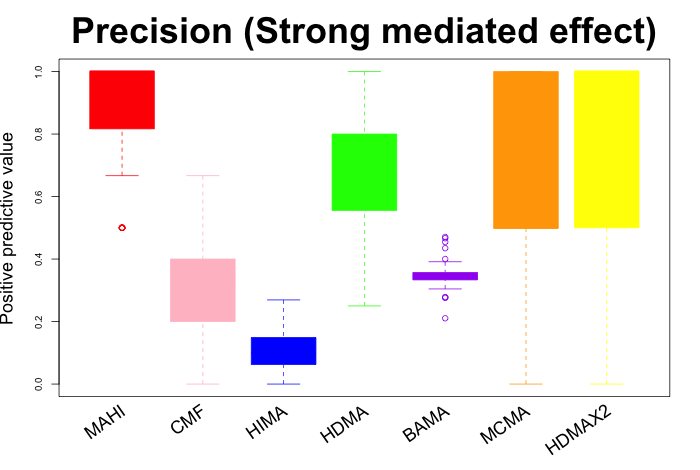}
\includegraphics[height=0.24\textwidth,width=0.3\textwidth]{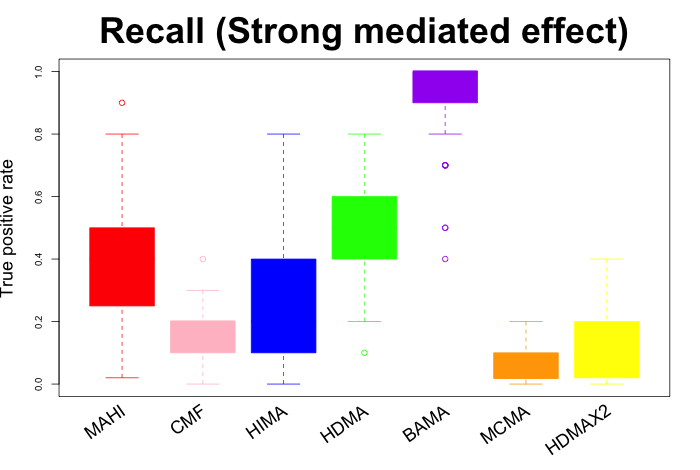}
\includegraphics[height=0.24\textwidth,width=0.3\textwidth]{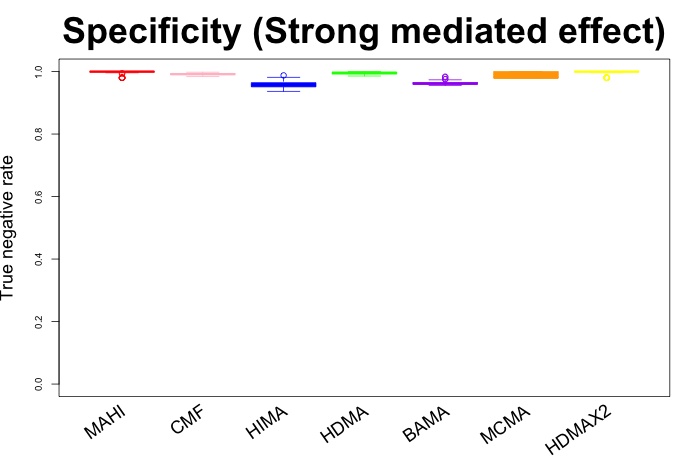}
\includegraphics[height=0.24\textwidth,width=0.3\textwidth]{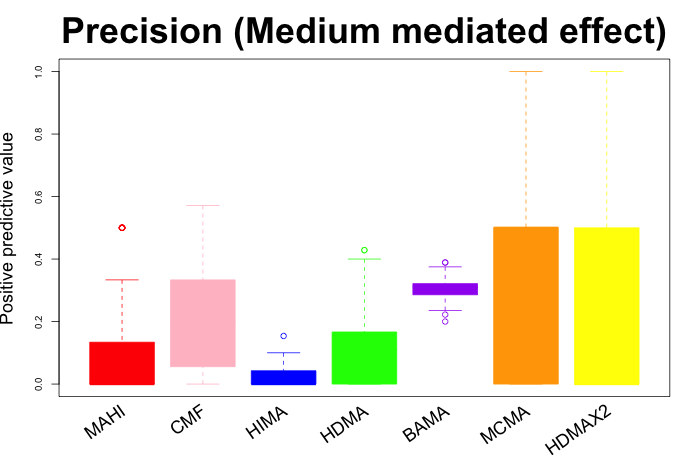}
\includegraphics[height=0.24\textwidth,width=0.3\textwidth]{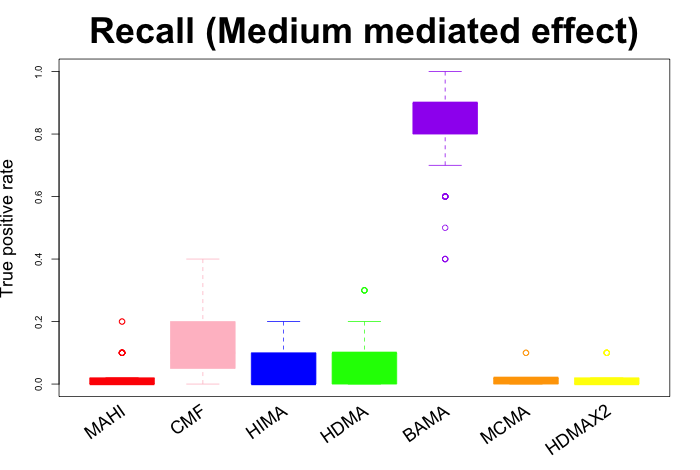}
\includegraphics[height=0.24\textwidth,width=0.3\textwidth]{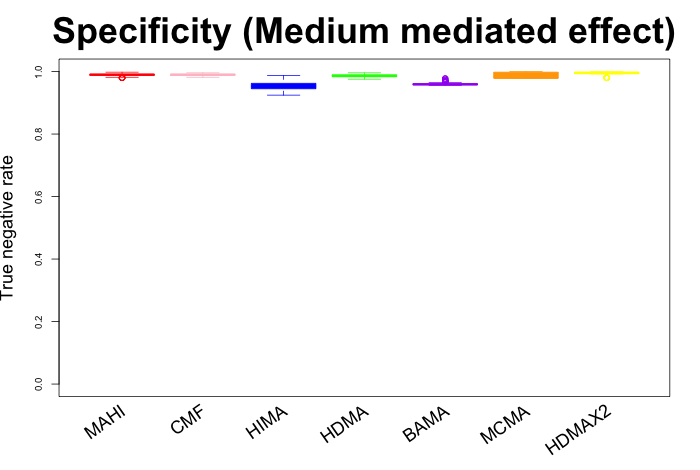}
\includegraphics[height=0.24\textwidth,width=0.3\textwidth]{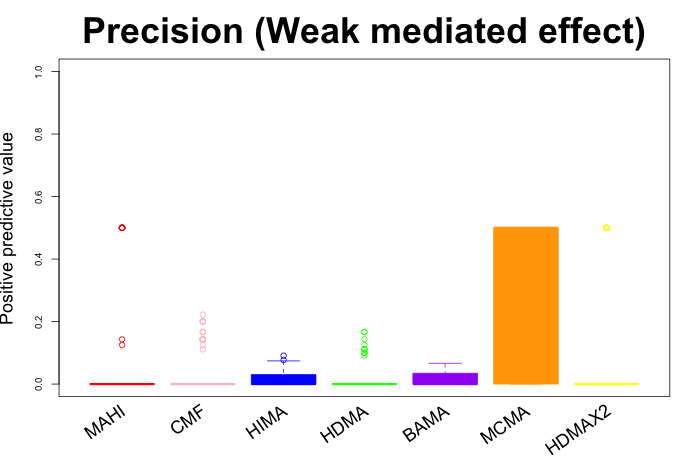}
\includegraphics[height=0.24\textwidth,width=0.3\textwidth]{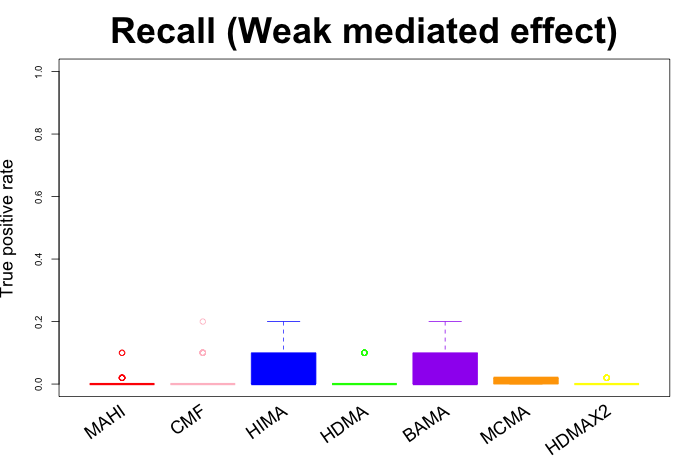}
\includegraphics[height=0.24\textwidth,width=0.3\textwidth]{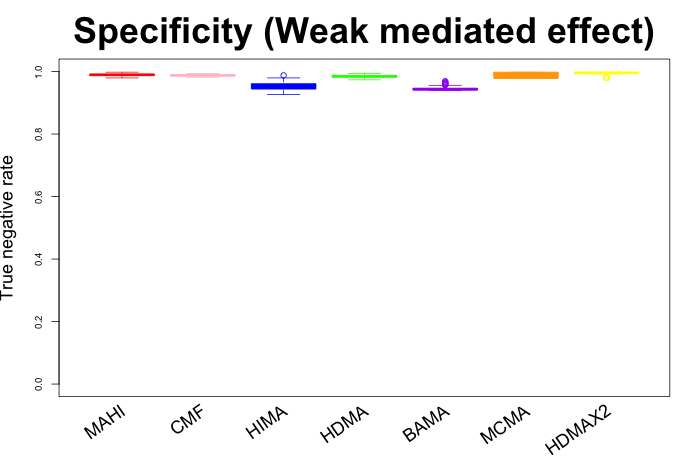}
\caption{Comparison of high-dimensional mediation analysis methods with regards to the ability to select the \textit{true} mediators $M_1,\ldots,M_{30}$. The results are displayed in the form of boxplots showing the distribution over 100 replicates simulated with model \eqref{grand_simed} for \textbf{continuous} outcomes. Variables $M_1,\ldots,M_{10}$ are \textit{strong} mediators, $M_{11},\ldots,M_{20}$ \textit{mild} mediators with \textit{medium} mediated effects, and $M_{21},\ldots,M_{30}$ \textit{weak} mediators. All the candidate mediators are \textbf{independent}. }\label{res_true_normal} 
\end{figure}

\begin{figure}[htp]
\vspace{-.5cm}
\centering
\includegraphics[height=0.24\textwidth,width=0.3\textwidth]{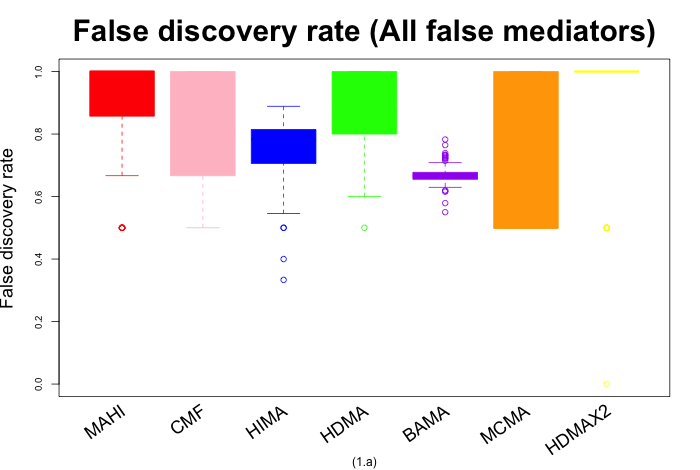}
\includegraphics[height=0.24\textwidth,width=0.3\textwidth]{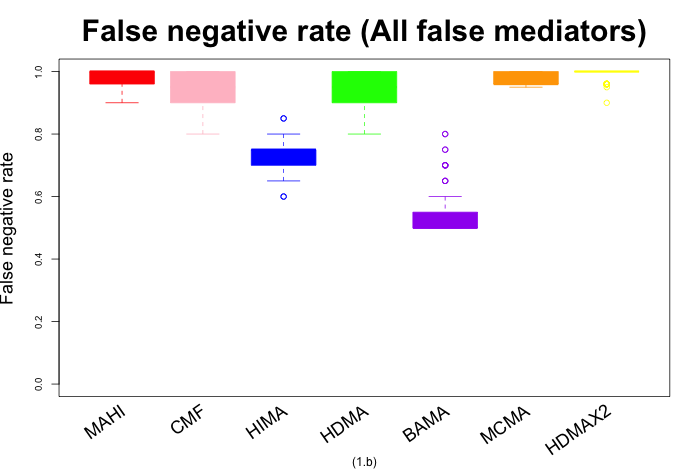}
\includegraphics[height=0.24\textwidth,width=0.3\textwidth]{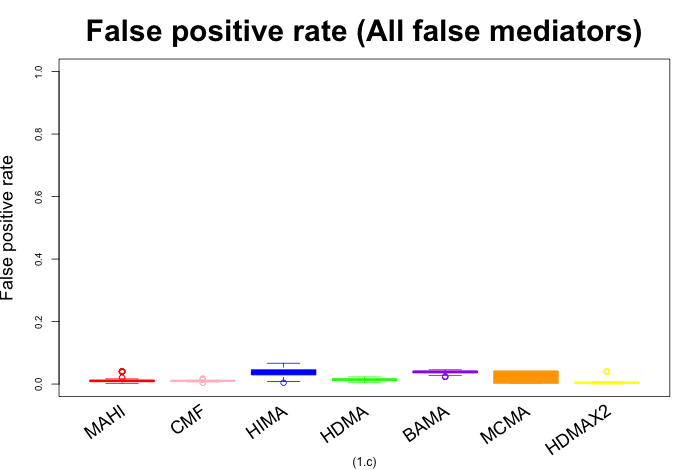}
\includegraphics[height=0.24\textwidth,width=0.3\textwidth]{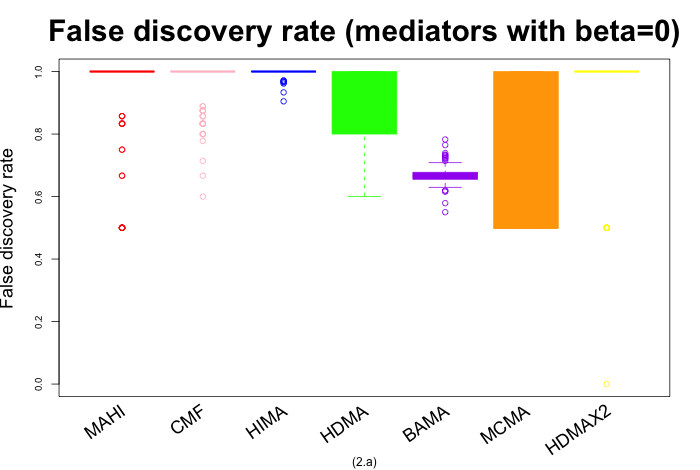}
\includegraphics[height=0.24\textwidth,width=0.3\textwidth]{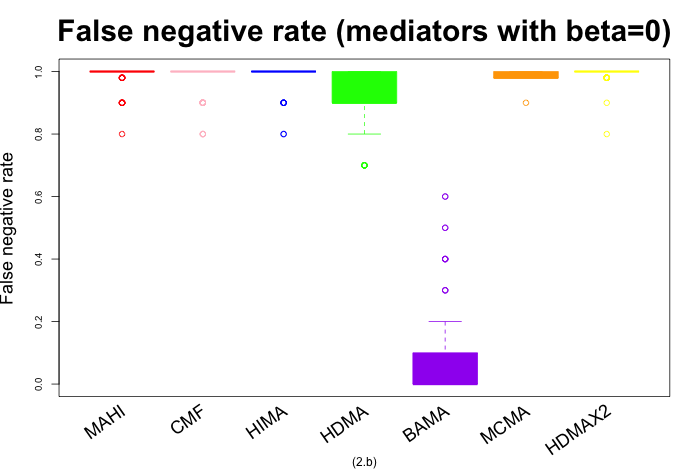}
\includegraphics[height=0.24\textwidth,width=0.3\textwidth]{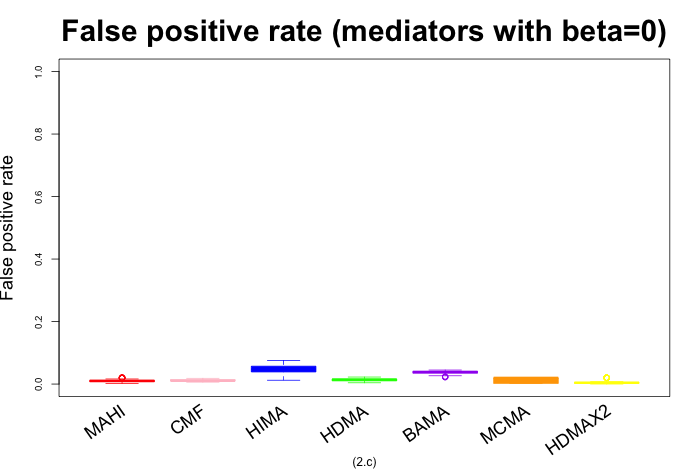}
\includegraphics[height=0.24\textwidth,width=0.3\textwidth]{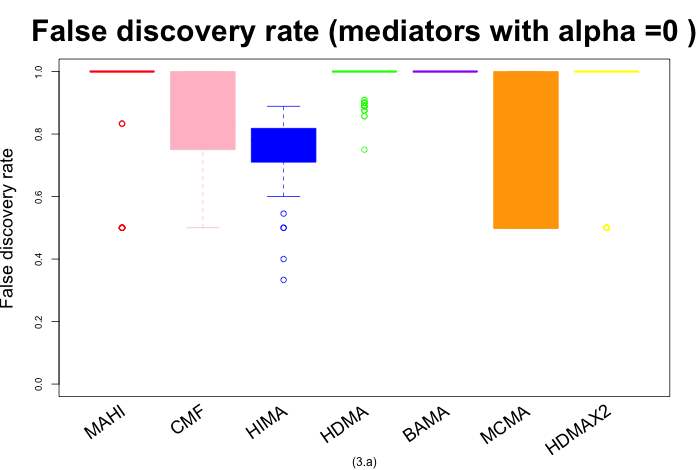}
\includegraphics[height=0.24\textwidth,width=0.3\textwidth]{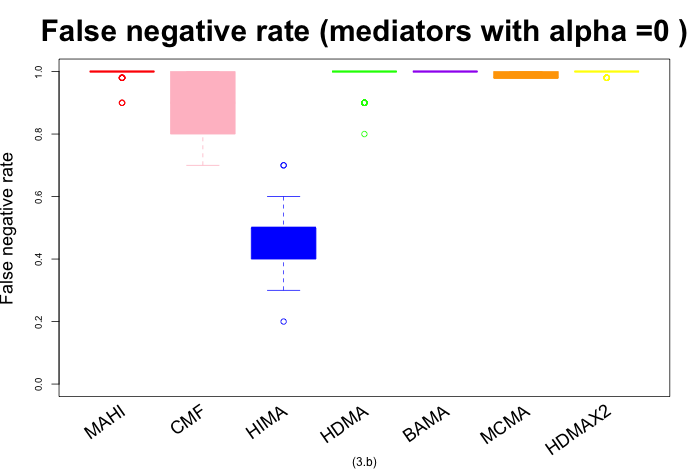}
\includegraphics[height=0.24\textwidth,width=0.3\textwidth]{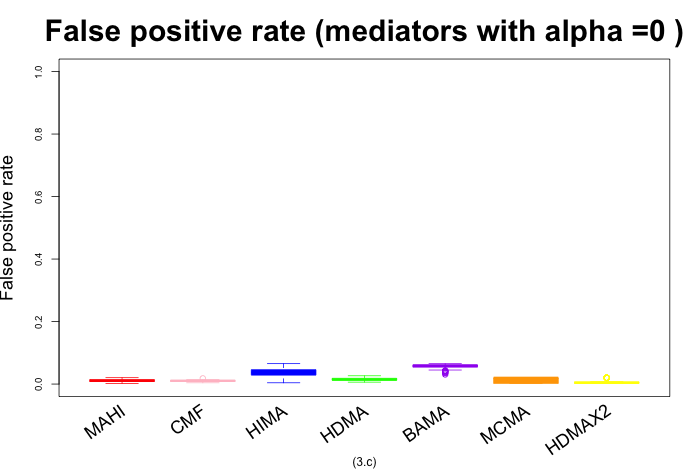}
\caption{Comparison of high-dimensional mediation analysis methods with regards to the selection of \textit{false} mediators (variables $M_{31},\ldots,M_{50}$). The results are displayed in the form of boxplots showing the distribution over 100 replicates simulated with model \eqref{grand_simed} for \textbf{continuous} outcomes. All the mediators are \textbf{independent}.}\label{res_fake_normal} 
\end{figure}

\begin{figure}[htp]
\vspace{-.5cm}
\centering
\includegraphics[height=0.24\textwidth,width=0.3\textwidth]{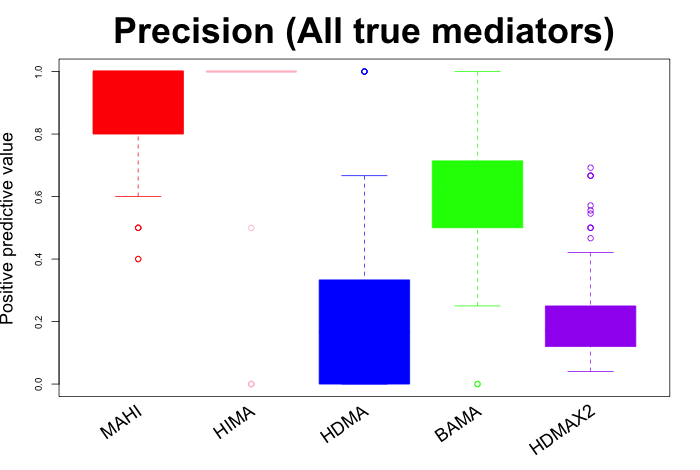}
\includegraphics[height=0.24\textwidth,width=0.3\textwidth]{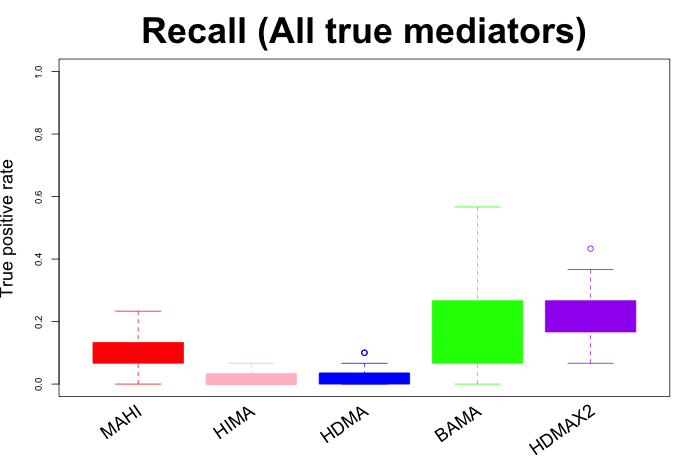}
\includegraphics[height=0.24\textwidth,width=0.3\textwidth]{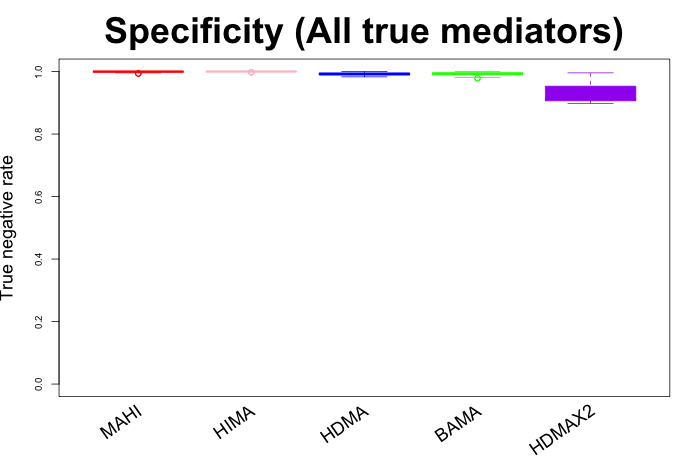}
\includegraphics[height=0.24\textwidth,width=0.3\textwidth]{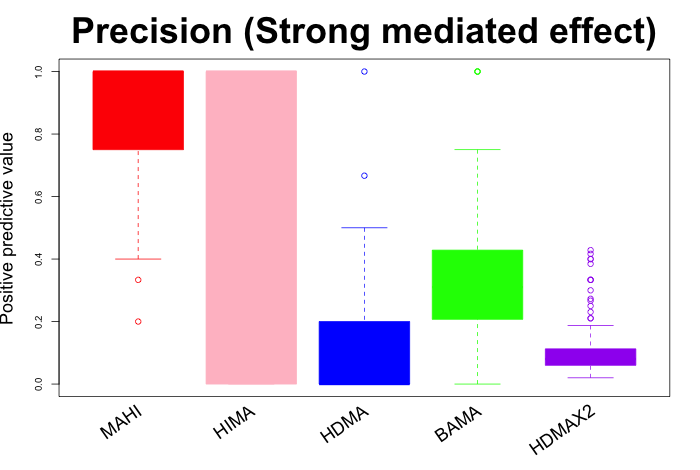}
\includegraphics[height=0.24\textwidth,width=0.3\textwidth]{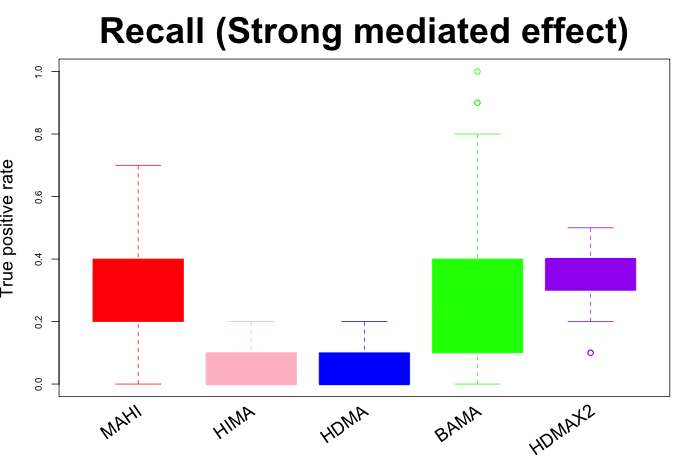}
\includegraphics[height=0.24\textwidth,width=0.3\textwidth]{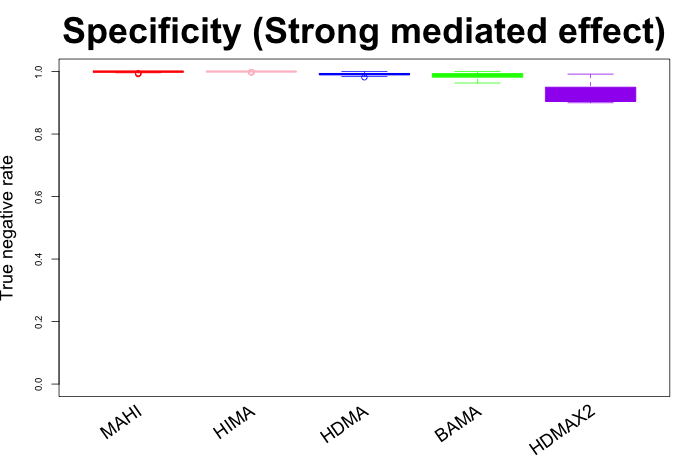}
\includegraphics[height=0.24\textwidth,width=0.3\textwidth]{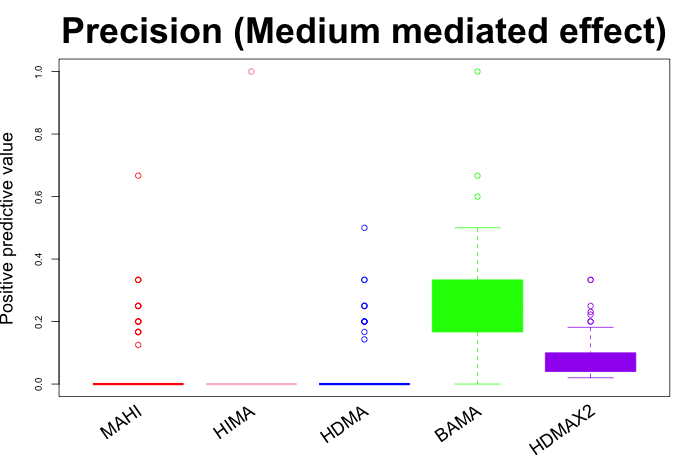}
\includegraphics[height=0.24\textwidth,width=0.3\textwidth]{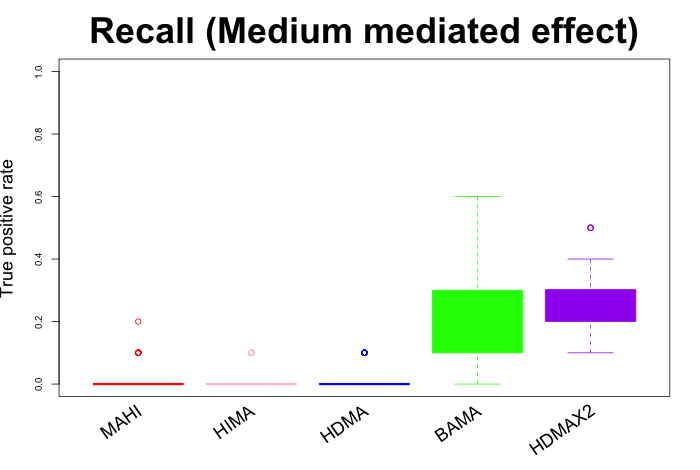}
\includegraphics[height=0.24\textwidth,width=0.3\textwidth]{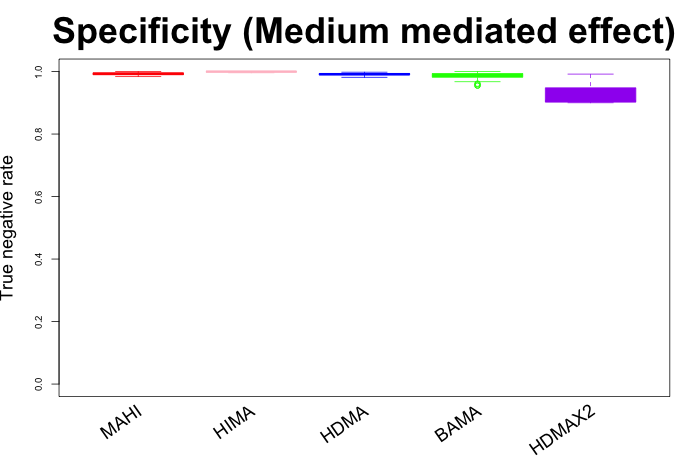}
\includegraphics[height=0.24\textwidth,width=0.3\textwidth]{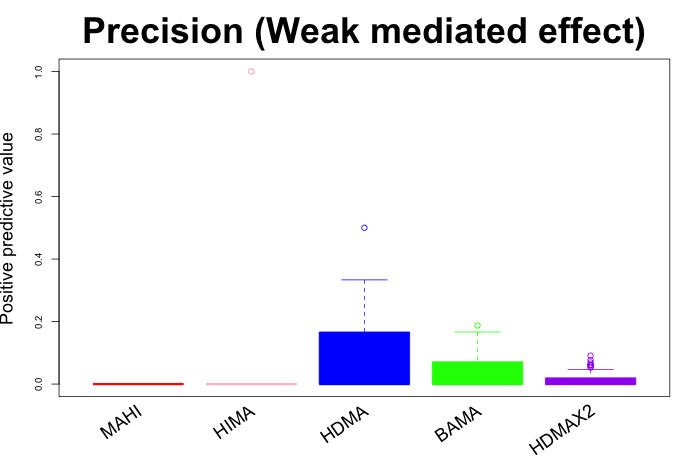}
\includegraphics[height=0.24\textwidth,width=0.3\textwidth]{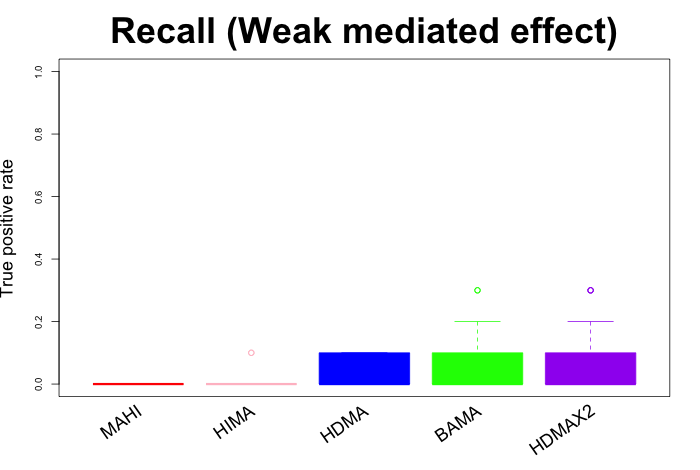}
\includegraphics[height=0.24\textwidth,width=0.3\textwidth]{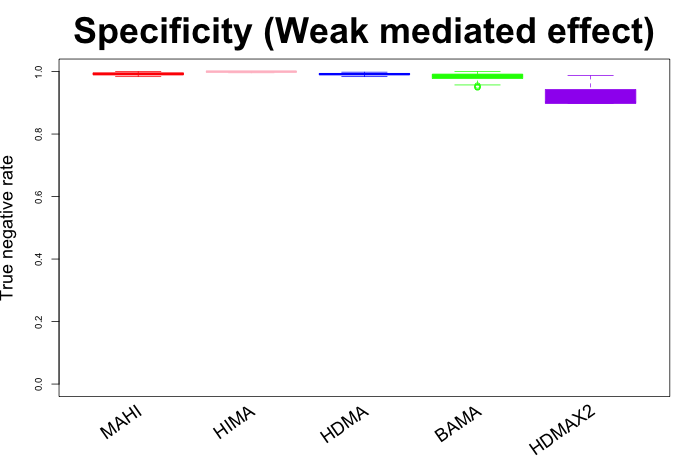}
\caption{Comparison of high-dimensional mediation analysis methods with regards to the ability to select the \textit{true} mediators $M_1,\ldots,M_{30}$. The results are displayed in the form of boxplots showing the distribution over 100 replicates simulated with model \eqref{grand_simed} for \textbf{continuous} outcomes including covariates. Variables $M_1,\ldots,M_{10}$ are \textit{strong} mediators, $M_{11},\ldots,M_{20}$ \textit{mild} mediators with \textit{medium} mediated effects, and $M_{21},\ldots,M_{30}$ \textit{weak} mediators. All the candidate mediators are \textbf{correlated}. }\label{res_true_normal_corr} 
\end{figure}

\begin{figure}[htp]
\vspace{-.5cm}
\centering
\includegraphics[height=0.24\textwidth,width=0.3\textwidth]{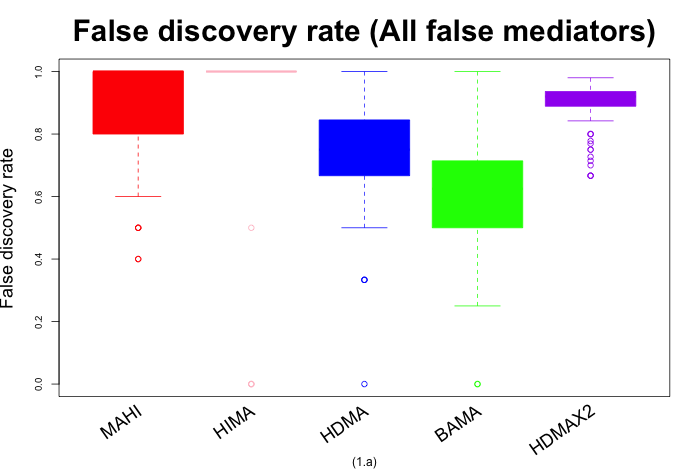}
\includegraphics[height=0.24\textwidth,width=0.3\textwidth]{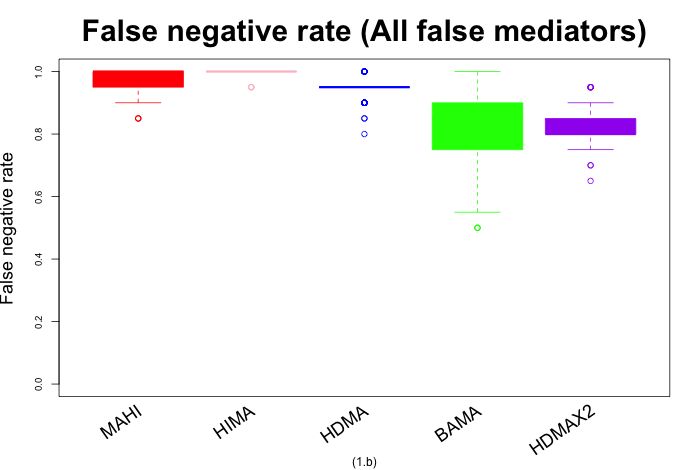}
\includegraphics[height=0.24\textwidth,width=0.3\textwidth]{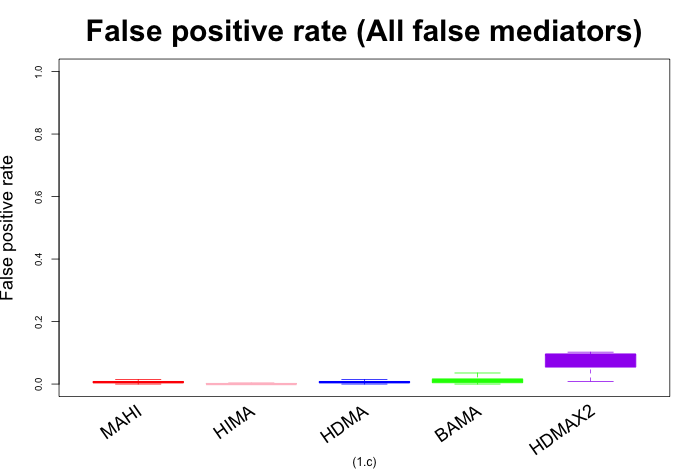}
\includegraphics[height=0.24\textwidth,width=0.3\textwidth]{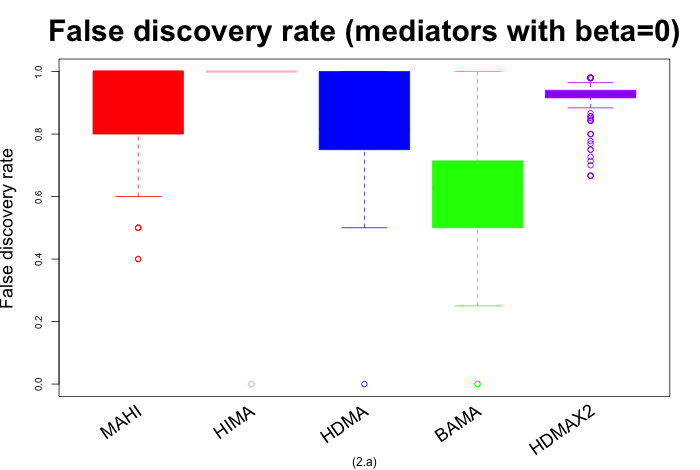}
\includegraphics[height=0.24\textwidth,width=0.3\textwidth]{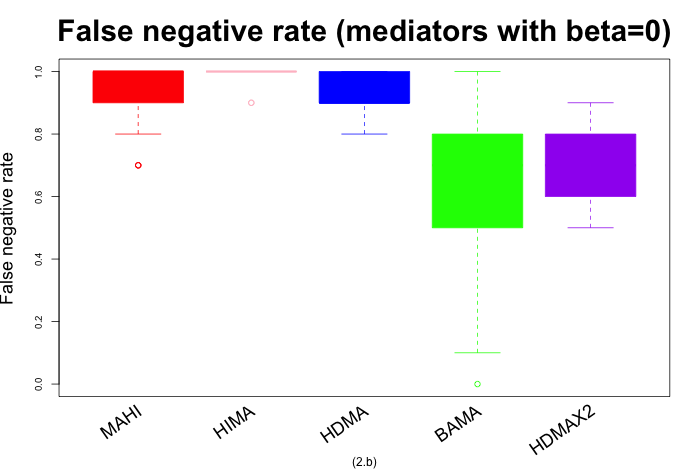}
\includegraphics[height=0.24\textwidth,width=0.3\textwidth]{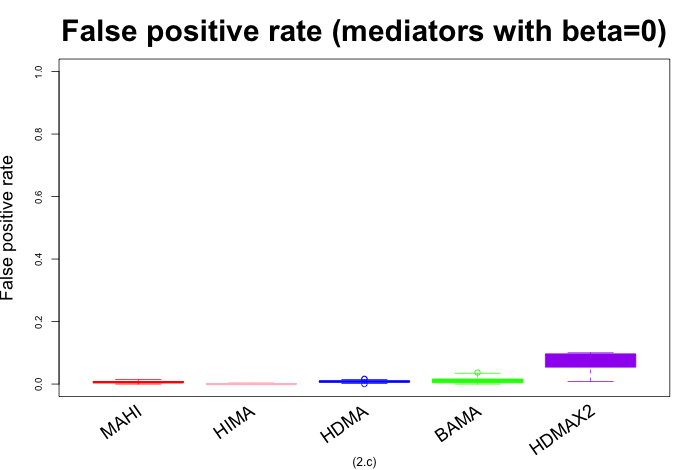}
\includegraphics[height=0.24\textwidth,width=0.3\textwidth]{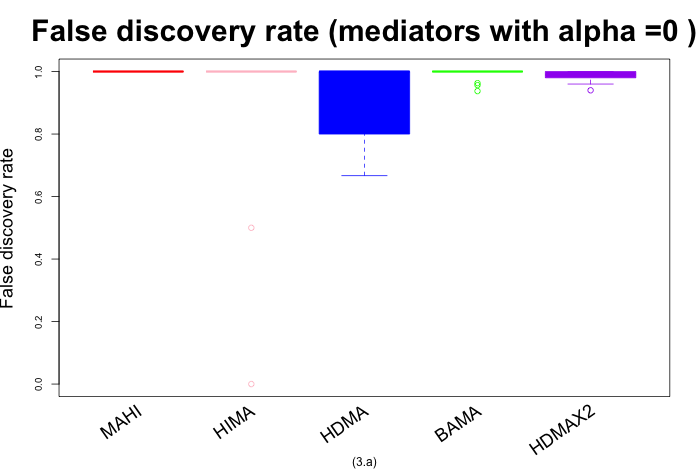}
\includegraphics[height=0.24\textwidth,width=0.3\textwidth]{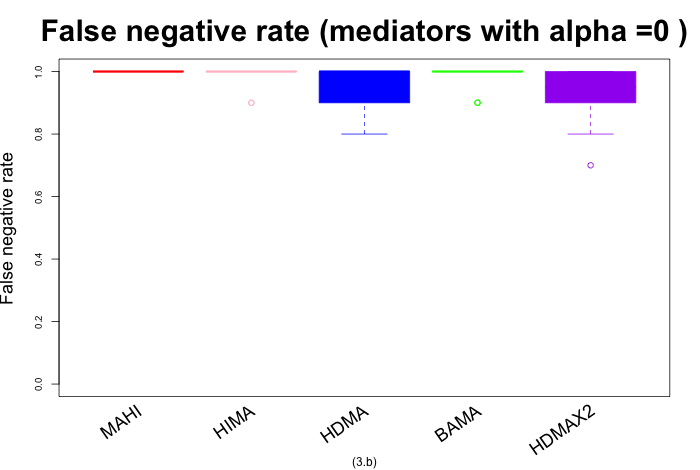}
\includegraphics[height=0.24\textwidth,width=0.3\textwidth]{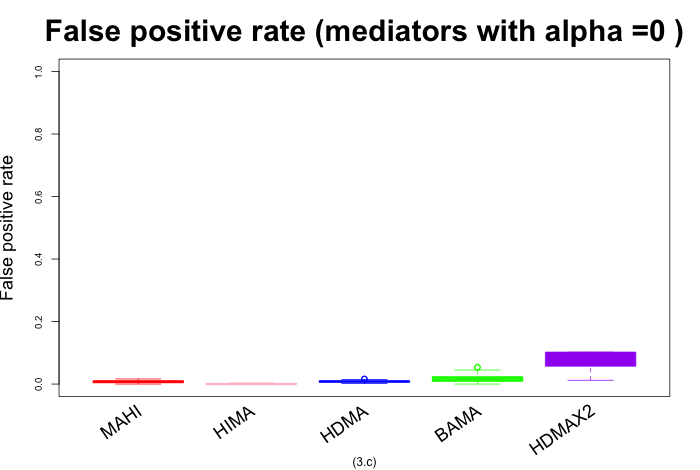}
\caption{Comparison of high-dimensional mediation analysis methods with regards to the selection of \textit{false} mediators (variables $M_{31},\ldots,M_{50}$). The results are displayed in the form of boxplots showing the distribution over 100 replicates simulated with model \eqref{grand_simed} for \textbf{continuous} outcomes. All the mediators are \textbf{correlated}.}\label{res_fake_normal_corr} 
\end{figure}

\begin{figure}[htp]
\vspace{-.5cm}
\centering
\includegraphics[height=0.24\textwidth,width=0.3\textwidth]{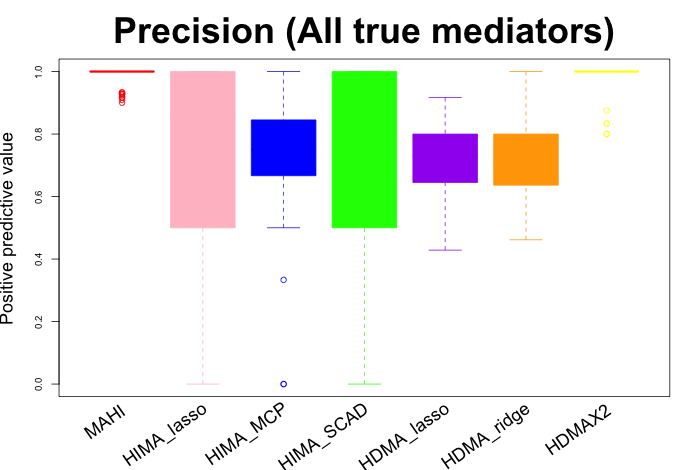}
\includegraphics[height=0.24\textwidth,width=0.3\textwidth]{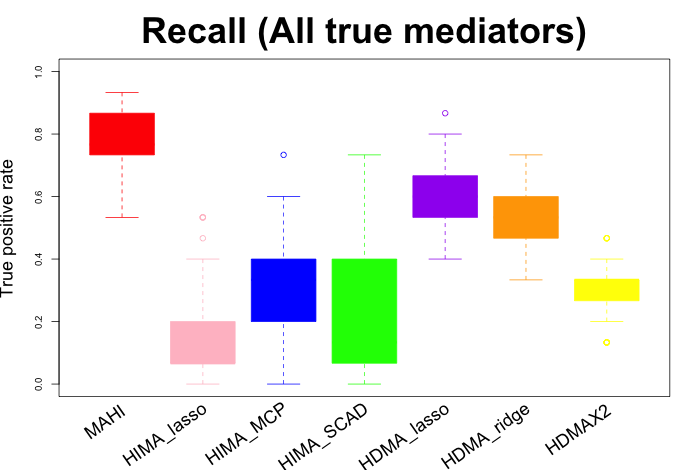}
\includegraphics[height=0.24\textwidth,width=0.3\textwidth]{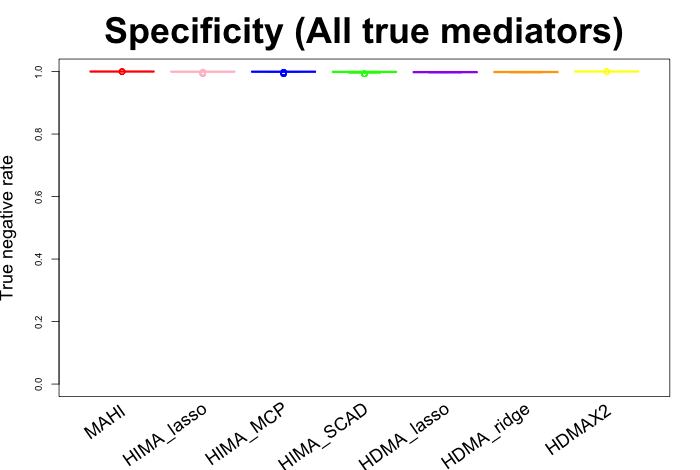}
\includegraphics[height=0.24\textwidth,width=0.3\textwidth]{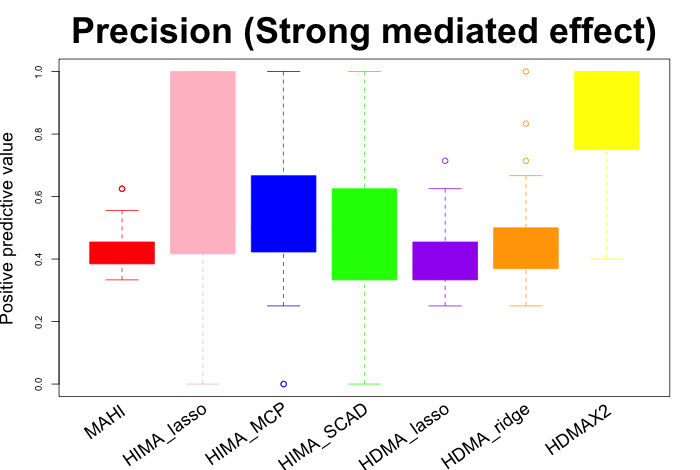}
\includegraphics[height=0.24\textwidth,width=0.3\textwidth]{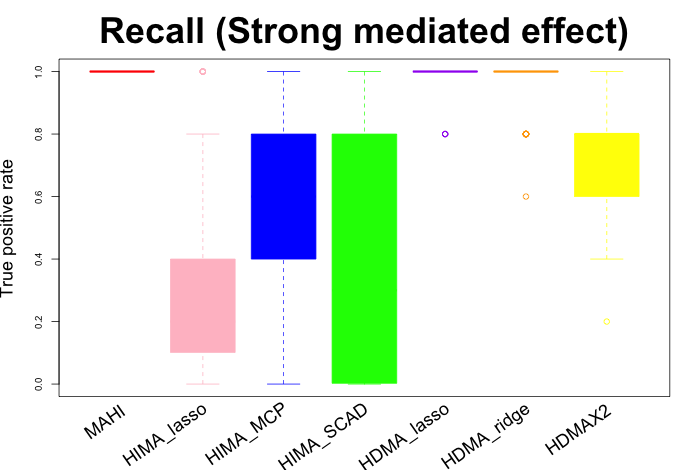}
\includegraphics[height=0.24\textwidth,width=0.3\textwidth]{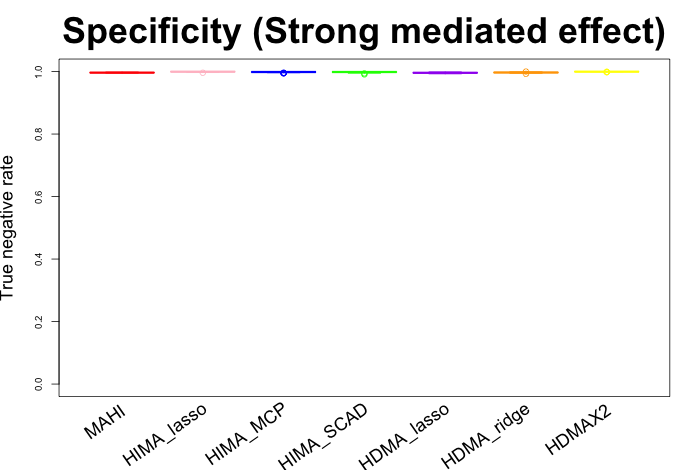}
\includegraphics[height=0.24\textwidth,width=0.3\textwidth]{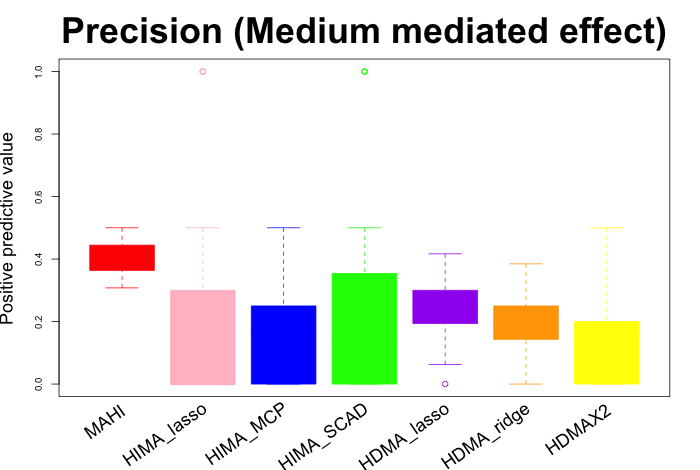}
\includegraphics[height=0.24\textwidth,width=0.3\textwidth]{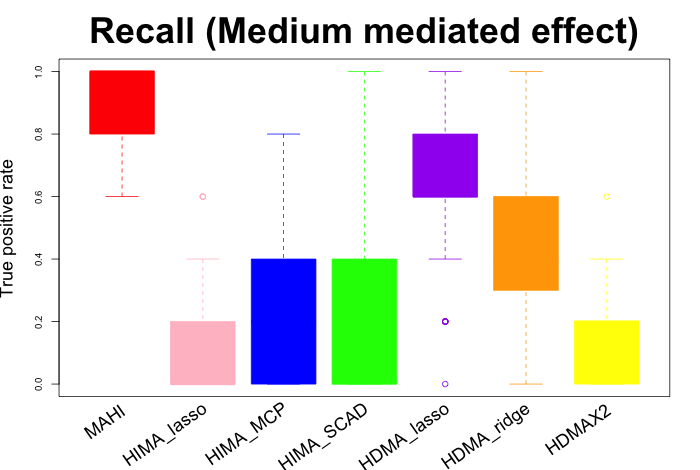}
\includegraphics[height=0.24\textwidth,width=0.3\textwidth]{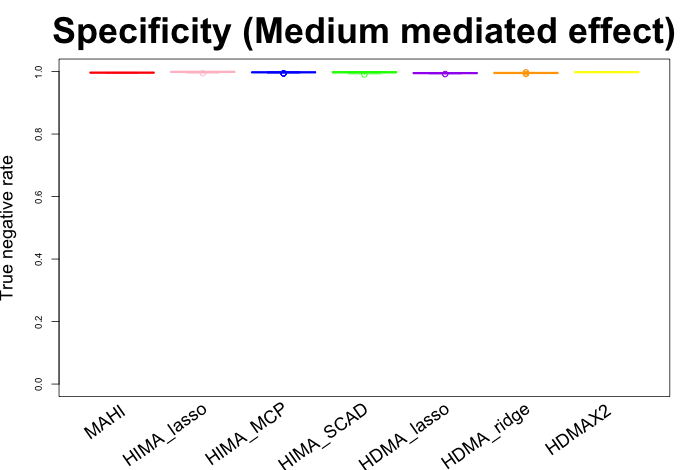}
\includegraphics[height=0.24\textwidth,width=0.3\textwidth]{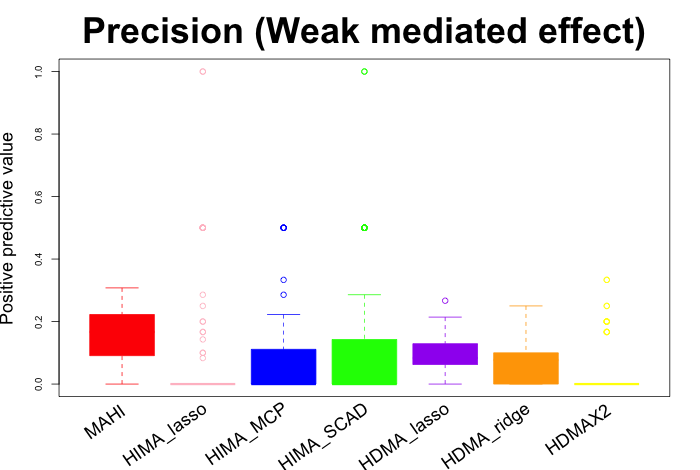}
\includegraphics[height=0.24\textwidth,width=0.3\textwidth]{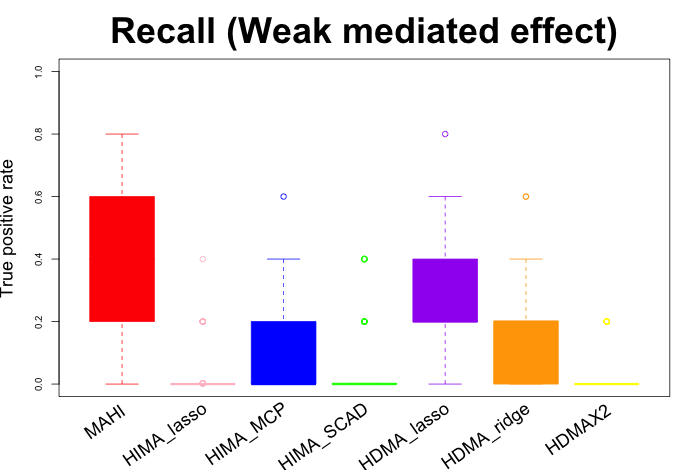}
\includegraphics[height=0.24\textwidth,width=0.3\textwidth]{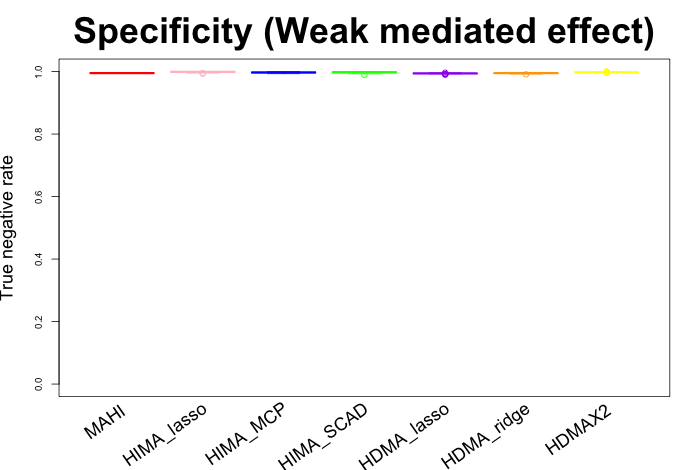}
\caption{Comparison of high-dimensional mediation analysis methods with regards to the ability to select \textit{true} mediators (variables $M_1,\ldots,M_{15}$). The results are displayed in the form of boxplots showing the distribution over 100 replicates with \textbf{binary} outcomes simulated with model \eqref{bin_grand_simed}. }\label{res_true_binary} 
\end{figure}

\begin{figure}[htp]
\vspace{-.5cm}
\centering
\includegraphics[height=0.24\textwidth,width=0.3\textwidth]{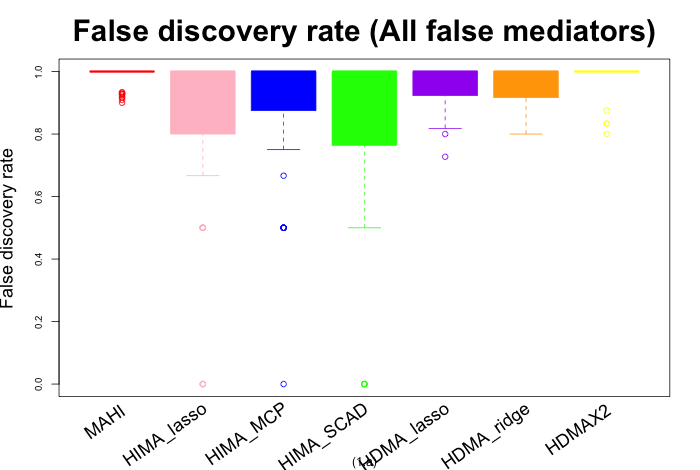}
\includegraphics[height=0.24\textwidth,width=0.3\textwidth]{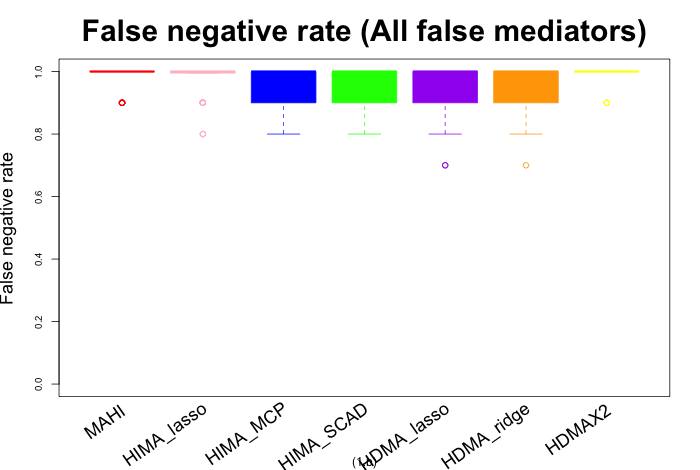}
\includegraphics[height=0.24\textwidth,width=0.3\textwidth]{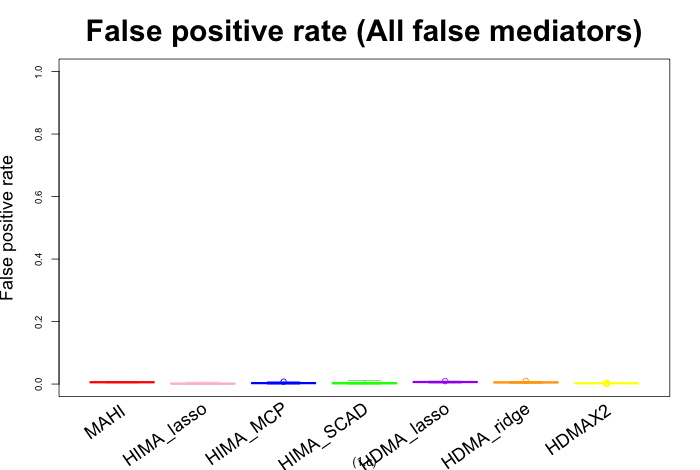}
\includegraphics[height=0.24\textwidth,width=0.3\textwidth]{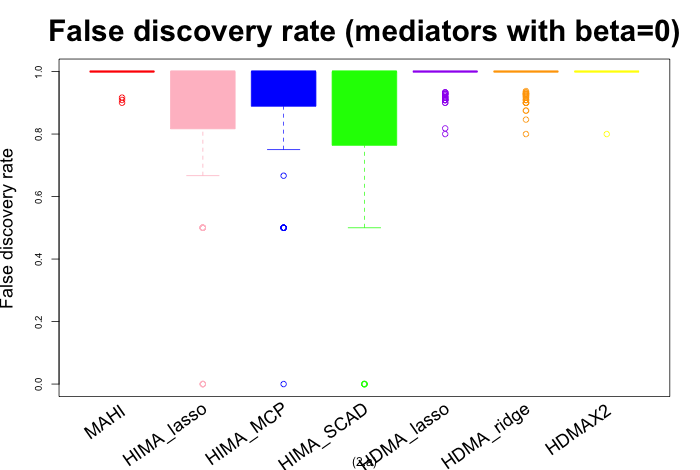}
\includegraphics[height=0.24\textwidth,width=0.3\textwidth]{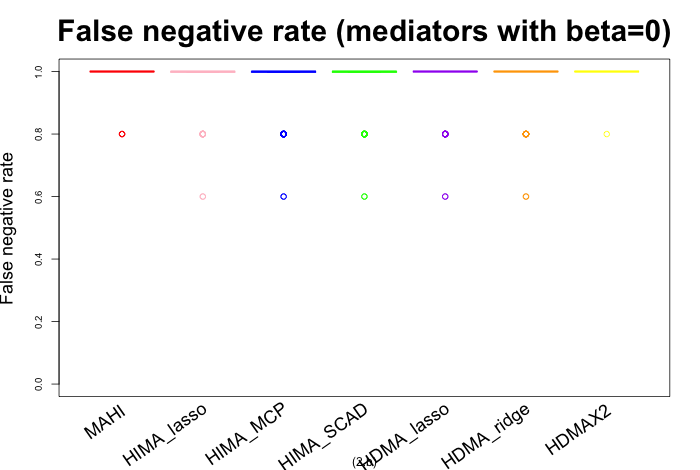}
\includegraphics[height=0.24\textwidth,width=0.3\textwidth]{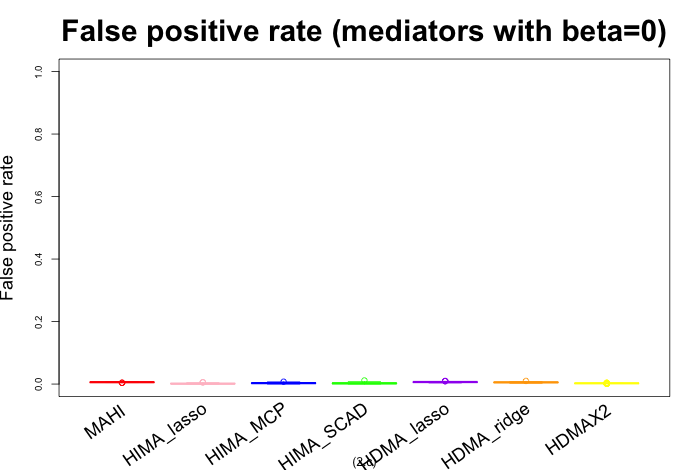}
\includegraphics[height=0.24\textwidth,width=0.3\textwidth]{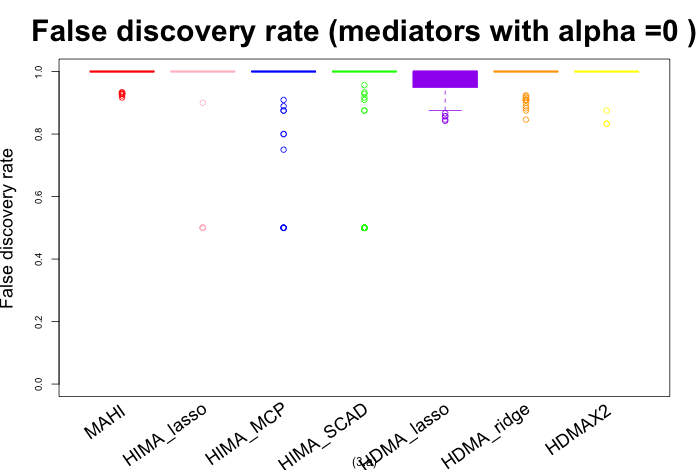}
\includegraphics[height=0.24\textwidth,width=0.3\textwidth]{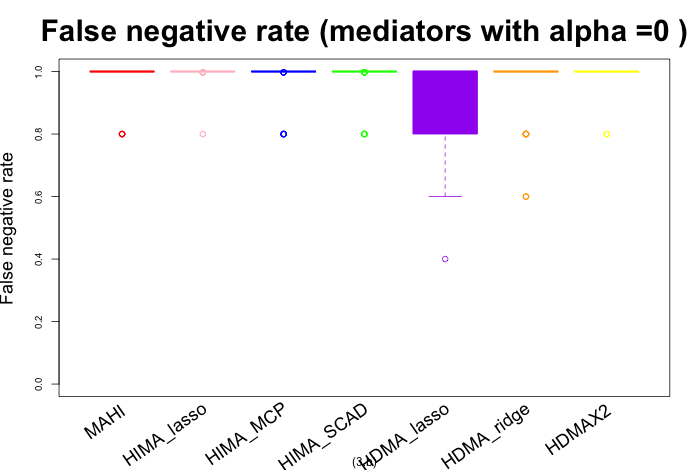}
\includegraphics[height=0.24\textwidth,width=0.3\textwidth]{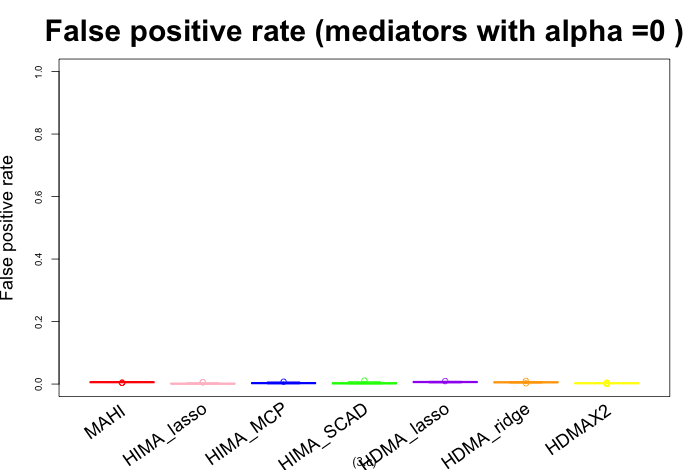}
\caption{Comparison of high-dimensional mediation analysis methods with regards to the selection of \textit{false} mediators $M_{16},\ldots,M_{25}$. The results are displayed in the form of boxplots showing the distribution over 100 replicates with \textbf{binary} outcomes simulated with model \eqref{bin_grand_simed}. }\label{res_fake_binary} 
\end{figure}

\begin{table}[htp]
\centering
\begin{tabular}{|c|c|c|c|c|c|c|c|}
\hline
\makecell{\textbf{Group of} \\ \textbf{candidate} \\ \textbf{mediators}} & \makecell{\textbf{Average} \\ $\delta^k$} & \makecell{\textbf{Average} \\  $|\delta^k|$} & \textbf{Step 1} & \textbf{Step 2} & \makecell{\textbf{Average} \\ $\hat{\delta}_k$} & \textbf{Bias} & \makecell{\textbf{Coverage} \\ \textbf{proportion}} \\
\hline
\multicolumn{8}{|c|}{\textbf{Independent candidate mediators, model without covariates}}\\
\hline
\makecell{10 \text{strong} \\ \text{true mediators}} & 57.60 & 120.6 & 100 & 38.7 & 55.57 & -2.03 & 0.95 \\
\hline
\makecell{10 \text{mild} \\ \text{true mediators}} & -3 & 5.60 & 57.3 & 1.8 & -2.36 & 0.51 & 0.97 \\
\hline
\makecell{10 \text{weak} \\ \text{true mediators}} & 0 & 0.25 & 24.3 & 0.2 & 0.07 & 0.03 & 0.92 \\
\hline
\makecell{5 false mediators \\ $(\alpha_{1,k},\beta_k)=(20,0)$} & 0 & 0 & 100 & 1.4 & 0.63 & 0.63 & 0.94 \\
\hline
\makecell{5 false mediators \\ $(\alpha_{1,k},\beta_k)=(4,0)$} & 0 & 0 & 93.4 & 1.2 & 0.99 & 0.99 & 0.94 \\
\hline
\makecell{5 false mediators \\$(\alpha_{1,k},\beta_k)=(0,20)$} & 0 & 0 & 33.2 & 0.6 & 0.88 & 0.88 & 0.96 \\
\hline
\makecell{5 false mediators \\$(\alpha_{1,k},\beta_k)=(0,4)$} & 0 & 0 & 0.8 & 0 & 0.03 & 0.03 & 1 \\
\hline
\makecell{450 false mediators \\$(\alpha_{1,k},\beta_k)=(0,0)$} & 0 & 0 & 0.10 & 0 & 0.52 & 0.52 & 0.99 \\
\hline
\multicolumn{8}{|c|}{\textbf{Correlated candidate mediators, model with covariates}} \\
\hline 
\makecell{10 \text{strong} \\ \text{true mediators}} & 57.6 & 120.6 & 95.4 & 30.6 & 59.30 & 1.70 & 0.82 \\
\hline
\makecell{10 \text{mild} \\ \text{true mediators}} & -3 & 5.6 & 73.2 & 1.8 & -2.64 & 0.36 & 0.94  \\
\hline
\makecell{10 \text{weak} \\ \text{true mediators}} & 0 & 0.25 & 64.1 & 0 & 0.40 & 0.40 & 0.99 \\
\hline
\makecell{5 false mediators \\ $(\alpha_{1,k},\beta_k)=(20,0)$} & 0 & 0 & 87.2 & 8.4 & 103.25 & 103.25 & 0.78 \\
\hline
\makecell{5 false mediators \\ $(\alpha_{1,k},\beta_k)=(4,0)$} & 0 & 0 & 31.4 & 0.8 & 5.70 & 5.70 & 0.92 \\
\hline
\makecell{5 false mediators \\ $(\alpha_{1,k},\beta_k)=(0,20)$} & 0 & 0 & 8.6 & 0 & -1.69 & -1.69 & 0.96 \\
\hline
\makecell{5 false mediators \\ $(\alpha_{1,k},\beta_k)=(0,4)$} & 0 & 0 & 3.6 & 0 & -0.03 & -0.03 & 1 \\
\hline
\makecell{450 false mediators \\ $(\alpha_{1,k},\beta_k)=(0,0)$} & 0 & 0 & 0.04 & 0 & 1.72 & 1.72 & 1 \\
\hline
\end{tabular}
\caption{Additional results for the proposed MAHI method under the same simulation setting as in Table \ref{tab_true_normal}. For each candidate mediator $M_k$, the true indirect effect $\delta^k$ was calculated by simulating a very large population according to model (\ref{grand_simed}) with either independent candidate mediators and no covariates ($\xi_{1k} = \xi_{2k} = \psi_{1k} = \psi_{2k} = 0$ for all $k$) or correlated candidate mediators and covariates. The parameters $\delta^k$ and $|\delta^k|$ were then averaged over all candidate mediators within each group, where groups are defined by the type of pairs $(\alpha_{1k}, \beta_k)$ as in Table \ref{grand_param}. For the remaining columns, 100 replicates of size $n = 100$ were generated. For each $M_k$ we considered : the number of times it was selected by Step 1 and Step 2 across 100 replicates; the mean estimate of the indirect effect and its bias across 100 replicates; the proportion of times the 95\% confidence interval contained the true value of $\delta^k$ over 100 replicates. These values were subsequently averaged over all candidate mediators within each group. }
\label{tab:CP_ind}
\end{table}

\end{appendices}

\end{document}